\documentclass[british,english]{article}
\usepackage[T1]{fontenc}
\usepackage[latin9]{inputenc}
\usepackage{geometry}
\usepackage{array}
\usepackage{float}
\usepackage{bm}
\usepackage{amsthm}
\usepackage{amsmath}
\usepackage{amssymb}
\usepackage{graphicx}
\usepackage{esint}
\usepackage[authoryear]{natbib}

\makeatletter

\providecommand{\tabularnewline}{\\}
\floatstyle{ruled}
\newfloat{algorithm}{tbp}{loa}
\providecommand{\algorithmname}{Algorithm}
\floatname{algorithm}{\protect\algorithmname}

\theoremstyle{plain}
\newtheorem{thm}{\protect\theoremname}
  \theoremstyle{remark}
  \newtheorem*{rem*}{\protect\remarkname}
  \theoremstyle{plain}
  \newtheorem{lem}[thm]{\protect\lemmaname}

\makeatother

\usepackage{babel}
  \addto\captionsbritish{\renewcommand{\lemmaname}{Lemma}}
  \addto\captionsbritish{\renewcommand{\remarkname}{Remark}}
  \addto\captionsbritish{\renewcommand{\theoremname}{Theorem}}
  \addto\captionsenglish{\renewcommand{\lemmaname}{Lemma}}
  \addto\captionsenglish{\renewcommand{\remarkname}{Remark}}
  \addto\captionsenglish{\renewcommand{\theoremname}{Theorem}}
  \providecommand{\lemmaname}{Lemma}
  \providecommand{\remarkname}{Remark}
\addto\captionsbritish{\renewcommand{\algorithmname}{Algorithm}}
\providecommand{\theoremname}{Theorem}

\begin{document}
\global\long\def\bb{\bm{\beta}}
\global\long\def\bt{\bm{\theta}}
\global\long\def\bmu{\bm{\mu}}
\global\long\def\bl{\bm{\lambda}}
\global\long\def\be{\bm{\eta}}
\global\long\def\E{\mathbb{E}}
\global\long\def\by{\bm{y}}
\global\long\def\bys{\bm{y^{\star}}}
\global\long\def\bS{\bm{\Sigma}}
\global\long\def\N{\mathcal{N}}
\global\long\def\I{\mathds{1}}
\global\long\def\A{\mathcal{A}}
\global\long\def\Ml{\mathcal{M}_{l}}
\global\long\def\Q{\mathcal{Q}}

\title{Expectation Propagation in the large-data limit}

\author{Guillaume Dehaene, Simon Barthelm\'e}
\maketitle
\begin{abstract}
Expectation Propagation \citep{Minka:EP} is a widely successful algorithm
for variational inference. EP is an iterative algorithm used to approximate
complicated distributions, typically to find a Gaussian approximation
of posterior distributions. In many applications of this type, EP
performs extremely well. Surprisingly, despite its widespread use,
there are very few theoretical guarantees on Gaussian EP, and it is
quite poorly understood.

In order to analyze EP, we first introduce a variant of EP: averaged-EP
(aEP), which operates on a smaller parameter space. We then consider
aEP and EP in the limit of infinite data, where the overall contribution
of each likelihood term is small and where posteriors are almost Gaussian.
In this limit, we prove that the iterations of both aEP and EP are
simple: they behave like iterations of Newton's algorithm for finding
the mode of a function. We use this limit behavior to prove that EP
is asymptotically exact, and to obtain other insights into the dynamic
behavior of EP: for example, that it may diverge under poor initialization
exactly like Newton's method. EP is a simple algorithm to state, but
a difficult one to study. Our results should facilitate further research
into the theoretical properties of this important method.
\end{abstract}

\section*{Introduction}

Current practice in Bayesian statistics favors MCMC methods, but so-called
\emph{variational approximations} are gaining traction. In machine
learning, where time constraints are primary, they have long been
the favored method for Bayesian inference \citep{BishopPRML}. Variational
methods provide fast, deterministic approximations to arbitrary distributions.
Examples include mean-field methods \citep{WainwrightJordan:GraphModelsExpFamVarInf},
INLA (Integrated Nested Laplace Approximation, \citealp{Rue:INLA}),
and Expectation Propagation (EP). 

EP was introduced in \citet{Minka:EP} and has proved to be one of
the most durably popular methods in Bayesian machine learning. It
gives excellent results in important applications like Gaussian process
classification \citep{KussRasmussen:AssessingApproxInfGP,NickishRasmussen:ApproxGaussianProcClass}
and is used in a wide range of applications (e.g., \citealt{Jylanki:EPforNeuralNetworks,Jylanki:RobustGaussianProcessReg,GehreJin:EPforNonlinearInverseProb,Ridgway;PACBayesAUCClassificationScoring}
). Recently EP has been shown to work very well in certain difficult
likelihood-free settings \citep{BarthelmeChopin:EPforLikFreeInf},
and has even been advocated as a generic form of inference in large-data
problems \citep{Gelman:EPWayOfLife,Xu:DistributedBayesPostSampling},
since EP is easy to parallelize.

Most of the work on EP concerns applications, and focuses on making
the method work well in various settings. Why and when the method
should work remains somewhat of a mystery, and in this article we
aim to make progress in that direction. A few theoretical results
are available when the approximating family is a discrete distribution,
in which case EP is equivalent to Belief Propagation, a well-studied
algorithm \citep{WainwrightJordan:GraphModelsExpFamVarInf}. The typical
case in Bayesian inference is to use multivariate Gaussians as the
approximating family, but very little is known about that case: \citet{RibeiroOpper:EPWithFactDistr}
study a limit of EP for neural network models (the limit of infinitely
many weights) and \citet{Titterington:EMalgVarApproxEP} gives partial
results on mixture models in the large-data limit. Despite these efforts,
two aspects of EP's behavior have remained elusive: its dynamical
behavior (does the EP iteration converge on a \emph{fixed dataset}?)
and its large-data behavior (do fixed points of the iteration converge
to the target distribution in the limit of infinite data?). In this
work, we focus on the dynamical behavior of EP and show that it is
asymptotically equivalent to the behavior of Newton's method \citep{NocedalWright:NumericalOptim}.
This enables us to prove that EP is exact in the large-data limit:
if the posterior in the large-data limit tends to a Gaussian (as they
usually do), then EP recovers the limiting Gaussian. Furthermore,
we show that on multimodal distributions, EP often has one fixed-point
for each mode. This also yields insights into why EP iterations can
be so unstable. 

The outline of the paper is as follows. In section \ref{sec:From-classic-EP-to-aEP},
we give a quick introduction to EP and introduce a simpler variant
which we call averaged-EP (aEP). aEP is mathematically simpler than
EP because it iterates over a much smaller parameter space (independent
of $n$, the number of data points), which makes our results easier
to state and to understand. We then present our theoretical contributions
in section \ref{sec:Asymptotically-exact}. Our main result concerns
the asymptotic behavior of the EP update, which turns out to be extremely
simple. This asymptotic behavior has many consequences, of which we
highlight two. First, EP and aEP asymptote to Newton's algorithm.
Second, EP is asymptotically exact, or more specifically the target
distribution and one specific EP fixed-point converge in total-variation
distance. In section \ref{sec:Consequences-following-from the limit behavior},
we then show that this Newton limit behavior of EP can give us some
intuition into how the iterations of the algorithm work. Finally,
in section \ref{sec:Conclusion} we discuss limitations of our results
and give directions for future work.

\subsection*{Notation and background\label{sub:Notation-and-background}}

Vectors are in bold, matrices are in bold and capitalized. Given a
multivariate function $f(\mathbf{x})$, we note $\nabla f$ its gradient
and $Hf$ its Hessian, the matrix of the second derivatives. Univariate
Gaussian distributions are represented as $\N(x|\mu,v)\propto\exp\left(-\frac{1}{2v}\left(x-\mu\right)^{2}\right)$,
although occasionally the exponential parameters $\beta=v^{-1},\, r=\beta\mu$
are used: $\N(x|r,\beta)\propto\exp\left(-\frac{1}{2}\beta x^{2}+rx\right)$.
We call $\beta$ the \emph{precision }and $r$ the \emph{linear shift}.
Table \ref{tab:An-EP-lexicon} provides a lexicon for EP and a summary
of the notation.

The goal of EP is to compute a Gaussian approximation of a target
distribution, which we note $p\left(\mathbf{x}\right)\propto\exp\left(-\psi\left(\mathbf{x}\right)\right)$.
This distribution factorizes into $n$ factor-functions (sites in
EP terminology): $p\left(\mathbf{x}\right)=\prod_{i=1}^{n}l_{i}\left(\mathbf{x}\right)$.
We note $\phi_{i}\left(\mathbf{x}\right)=-\log\left(l_{i}\left(\mathbf{x}\right)\right)$.
EP produces a Gaussian approximation $q\left(\mathbf{x}\right)\approx p\left(\mathbf{x}\right)$
with the same factor structure, $n$ Gaussian factors $f_{i}\left(\mathbf{x}\right)$
such that: $q\left(\mathbf{x}\right)=\prod_{i=1}^{n}f_{i}\left(\mathbf{x}\right)$.
Each Gaussian factor $f_{i}\left(\mathbf{x}\right)$ approximates
the corresponding target factor $l_{i}\left(\mathbf{x}\right)$ .

\subsubsection*{Newton's algorithm as an approximate inference method}

Approximate inference methods aim to find a tractable approximation
$q(\mathbf{x})$ to a complicated density $p(\mathbf{x})$. Most of
them operate by solving:

\[
\underset{q\in\mathbb{\mathcal{Q}}}{\mbox{argmin}\,}D(p||q)
\]

where $\mathcal{Q}$ denotes some set of tractable distributions and
$D$ is a divergence measure. Depending on the choice of divergence
measure and approximating distribution, one can derive various variational
algorithms. These methods are often iterative and produce a sequence
of approximations $q_{1},\ldots,q_{T}$ that should hopefully tend
to a locally optimal approximation. 

One of our key results proves that, in the large-data limit (denoted
here by $n\rightarrow\infty$), EP behaves like Newton's algorithm
(NT, see e.g. \citealp{NocedalWright:NumericalOptim} for an introduction).
NT aims to find a mode of a target probability distribution $p(x)\propto\exp\left(-\psi\left(x\right)\right)$
through an iterative procedure. We present here the one-dimensional
version. Once initialized at a point $\mu_{1}$, a sequence of points
$(\mu_{t})$ is constructed with:

\begin{equation}
\mu_{t+1}=\mu_{t}-\left[\psi^{''}\left(\mu_{t}\right)\right]^{-1}\psi^{'}\left(\mu_{t}\right)\label{eq:GN iteration}
\end{equation}

This iteration can be viewed as a gradient descent with a Hessian
correction. It can also be viewed as approximating $\log(p)$ as its
second degree Taylor expansion around $\mu_{t}$, and then setting
$\mu_{t+1}$ as the extremum of that polynomial.

With a slight modification we can restate NT as an approximate inference
algorithm iterating on Gaussian approximations of $p$, which makes
the parallel to EP more obvious. Starting from an arbitrary Gaussian
$g_{1}$, with mean $\mu_{1}$, we construct a sequence of Gaussian
approximations $(g_{t})$ through iterating the following steps:
\begin{enumerate}
\item Compute $\delta r_{t\text{+1}}=-\psi^{'}\left(\mu_{t}\right)$ and
$\beta_{t+1}=\psi^{''}\left(\mu_{t}\right)$
\item Compute a Gaussian approximation to $p\left(x\right)$: $g_{t+1}\left(x\right)\propto\exp\left(\delta r_{t\text{+1}}\left(x-\mu_{t}\right)-\beta_{t+1}\frac{(x-\mu_{t})^{2}}{2}\right)$
\item Compute the mean of $g_{t\text{+1}}$: $\mu_{t+1}=\mu_{t}-\left[\psi^{''}\left(\mu_{t}\right)\right]^{-1}\psi^{'}\left(\mu_{t}\right)$
\end{enumerate}
With this change, the fixed point of NT is now the Gaussian distribution
$\N\left(x\middle|x^{\star},\left[\psi^{''}\left(x^{\star}\right)\right]^{-1}\right)$
centered at $x^{\star}$, the mode of $p$, and with precision the
Hessian of $\log(p)$ at the mode $x^{\star}$. Thus, the fixed point
of this NT variant is the canonical Gaussian approximation (CGA) at
the mode of $p$, also sometimes referred to as the ``Laplace''
approximation (which is erroneous as the Laplace approximation actually
refers to approximating integrals and not probability distributions).

An important issue is the convergence of NT. It has fast convergence
when initialized close to a mode of $p$. Technically, convergence
is quadratic, i.e. $|\mu_{t+1}-x^{\star}|\leq c(\mu_{t}-x^{\star})^{2}$.
However, that is only true in a neighborhood of the mode and the basic
version of the algorithm, which we presented here, does not generally
converge for all starting points $\mu_{1}$. In order to obtain an
algorithm with guaranteed convergence, one solution is to complement
NT with a line-search algorithm. As we shall see, EP can also have
unstable behavior when initialized too far from its fixed points:
we return to this important issue in section \ref{sub:Instability}.

\subsubsection*{Log-concave distributions and the Brascamp-Lieb theorem}

Our theoretical results depend on a very powerful theorem on log-concave
probability distributions, called the Brascamp-Lieb theorem \citep{Brascamp1976151,SaumardWellner:LogConcavityReview}.
Let $LC\left(\mathbf{x}\right)\propto\exp\left(-\psi\left(\mathbf{x}\right)\right)$
be a log-concave distribution (i.e., $H\psi\left(\mathbf{x}\right)$
is always symmetric positive definite). The variance of any statistic
$S\left(\mathbf{x}\right)$ is then bounded according to:

\begin{equation}
\mbox{var}_{LC}\left(S(\mathbf{x})\right)\leq E_{LC}\left(\left(\nabla S\right)^{T}\left[H\psi(\mathbf{x})\right]^{-1}\nabla S\right)\label{eq:Brascamp-Lieb}
\end{equation}

We use this result in the particular case $S\left(\mathbf{x}\right)=\mathbf{x}$
from which we get an upper-bound on the variance:

\begin{equation}
\mbox{var}_{LC}\left(\mathbf{x}\right)\leq E_{LC}\left(\left[H\psi\left(\mathbf{x}\right)\right]^{-1}\right)
\end{equation}

Further more, if the log-Hessian is lower-bounded (as a matrix inequality):$H\psi\left(\mathbf{x}\right)\geq\mathbf{B}_{m}$,
then the variance has an even simpler upper-bound:
\begin{equation}
\mbox{var}_{LC}\left(\mathbf{x}\right)\leq\mathbf{B}_{m}^{-1}
\end{equation}

\begin{table}
\begin{centering}
\begin{tabular*}{12cm}{@{\extracolsep{\fill}}|c|>{\raggedright}p{8cm}|}
\hline 
Term & Explanation\tabularnewline
\hline 
\hline 
Target distribution & The distribution we wish to approximate: $p(\mathbf{x})\propto\prod_{i=1}^{n}l_{i}\left(\mathbf{x}\right)$\tabularnewline
\hline 
EP approximation & An exponential-family distribution with the same factor structure
as $p(\mathbf{x})$, $q\left(\mathbf{x}\right)\propto\prod_{i=1}^{n}f_{i}\left(\mathbf{x}\right)=\frac{\exp\left(\sum\bl_{i}^{t}\mathbf{s}\left(\mathbf{x}\right)\right)}{Z\left(\sum\lambda_{i}\right)}$\tabularnewline
\hline 
``Site'' or ``factor'' & A factor $l_{i}\left(\mathbf{x}\right)$ in the target distribution\tabularnewline
\hline 
Site approximation & A factor $f_{i}\left(\mathbf{x}\right)$ in the approximation\tabularnewline
\hline 
Cavity prior & The approximate distribution with site $i$ taken out, i.e. $q_{-i}\left(\mathbf{x}\right)\propto\exp\left(\sum_{j\neq i}\bl_{j}^{t}\mathbf{s}\left(\mathbf{x}\right)\right)$
. In aEP, $q_{-}\left(\mathbf{x}\right)\propto q^{\frac{n-1}{n}}\left(\mathbf{x}\right)$
is independent of $i$. \tabularnewline
\hline 
Hybrid distribution & The product of a cavity prior and a true site, i.e.. $h_{i}\left(\mathbf{x}\right)\propto q_{-i}\left(\mathbf{x}\right)l_{i}\left(\mathbf{x}\right)$\tabularnewline
\hline 
\end{tabular*}
\par\end{centering}

\protect\caption{An EP lexicon\label{tab:An-EP-lexicon}}
\end{table}

\section{From classic EP to averaged-EP (aEP)\label{sec:From-classic-EP-to-aEP}}

\subsection*{Classic EP}

In this section we introduce EP in the exponential-family notation
used by \citet{Seeger:EPExpFam}, because it is neat, generic and
compact. EP has been introduced from a variety of viewpoints, and
the versions given in \citet{Minka:DivMeasuresMP,Seeger:EPExpFam,BishopPRML,Raymond:ExpectationPropagation}
are all potentially useful. 

Following \citet{Minka:DivMeasuresMP}, given a target distribution
$p(\mathbf{x})$, EP aims to solve

\begin{equation}
\underset{q\in\mathbb{\mathcal{Q}}}{\mbox{argmin}}\, KL(p||q)\label{eq:KL-optimisation}
\end{equation}

where $\mathcal{Q}$ is an approximating family and $KL$ denotes
the Kullback-Leibler divergence. Here we focus on the Gaussian case
but other exponential families may be used (for example, the Gaussian-Wishart
family is used in \citealp{Paquet:PerturbationCorrectionsApproxInference}). 

A central aspect of EP is that it relies on a \emph{factorization
}of $p$, i.e. that the posterior decomposes into a product of terms:

\begin{equation}
p(\mathbf{x})=\frac{1}{Z}\prod_{i=1}^{n}l_{i}\left(\mathbf{x}\right)\label{eq:true-posterior}
\end{equation}

where usually one of the terms corresponds to the prior and the rest
to independent likelihood terms (here and elsewhere $Z=\int\prod_{i=1}^{n}l_{i}\left(\mathbf{x}\right)\mbox{d}\mathbf{x}$
is a normalization constant). The decomposition is non-unique and
the performance and feasibility of EP depend on the factorization
one picks. The approximation has the same factor structure:

\begin{equation}
q(\mathbf{x})\propto\prod_{i=1}^{n}f_{i}(\mathbf{x})\label{eq:q-factorised}
\end{equation}

Following Seeger, we call the $l_{i}$'s \emph{sites} and the corresponding
$f_{i}$'s \emph{site approximations}. The site approximations have
exponential-family form (e.g., Gaussian)

\[
f_{i}(\mathbf{x})=\exp\left(\bm{\lambda}_{i}^{t}\bm{t}\left(\mathbf{x}\right)\right)
\]

which the approximation inherits

\begin{equation}
q_{\bl_{s}}(\mathbf{x})=\exp\left\{ \bm{\lambda}_{s}\bm{t}\left(\mathbf{x}\right)-\phi\left(\bm{\lambda}_{s}\right)\right\} \label{eq:exponential-family}
\end{equation}

where $\bl_{s}=\sum\bl_{i}$. Note that $\bl_{s}$ represents the
so-called natural parameters for the approximation. According to a
well-known property of exponential families, the gradient of the partition
function, $\nabla\phi\left(\bl\right)$ returns the expected value
of the sufficient statistics for a given value of the natural parameters:
\[
\be=\nabla\phi\left(\bl\right)=\frac{1}{Z}\int\bm{t}\left(\mathbf{x}\right)\exp\left\{ \bl^{t}\bm{t}\left(\mathbf{x}\right)\right\} \mbox{d}\mathbf{x}
\]

Its inverse $\nabla\phi^{-1}$ transforms expected values of the sufficient
statistics into natural parameters.

A well known result for exponential families shows that the global
solution of problem \eqref{eq:KL-optimisation} is a moment-matching
solution:

\[
\be^{*}=E_{p}\left(\bm{t}\left(\mathbf{x}\right)\right)
\]

In the Gaussian case, what this means is that the best approximation
of $p$ according to KL divergence is a Gaussian with the same mean
and covariance. Of course, directly computing the mean and covariance
of $p$ is intractable, and so EP tries to get there by successive
refinements of an approximation. 

Specifically, EP tries to improve the approximation sequentially by
introducing \emph{hybrid }distributions which interpolate between
the current approximation and the true posterior. A hybrid distribution
$h_{i}$ contains one site from the \emph{true posterior, }but all
the rest come from the approximation:

\begin{equation}
h_{i}\left(\mathbf{x}\right)\propto q_{-i}\left(\mathbf{x}\right)l_{i}(\mathbf{x}),\quad q_{-i}(\mathbf{x})=\prod_{j\neq i}f_{j}(\mathbf{x}).\label{eq:hybrid}
\end{equation}

Hybrids should be tractable, meaning that one should be able to compute
their moments quickly. Note that in exponential-family notation, the
$q_{-i}\left(\mathbf{x}\right)$ distribution is simply:

\[
q_{-i}(\mathbf{x})\propto\exp\left\{ \left(\bl_{s}-\bl_{i}\right)\bm{t}\left(\mathbf{x}\right)\right\} 
\]

EP improves the approximation sequentially by (a) picking a site $i$
(b) computing the moments of the hybrid $h_{i}$ and (c) setting $\bl_{s}$
such that the moments of $q_{\bl_{s}}$ match the moments of the hybrid.

\begin{algorithm}
Loop until convergence

For $i$ in $1\ldots n$
\begin{enumerate}
\item Compute ``cavity'' parameter $\bl_{-i}\leftarrow\mathbf{\bl}_{s}-\bl_{i}$
\item Form hybrid distribution and compute its moments 
\[
\mathbf{\be}_{i}=\frac{1}{Z_{i}}\int\bm{t}\left(\mathbf{x}\right)l_{i}\left(\mathbf{x}\right)\exp\left\{ \bm{\lambda}_{-i}^{t}\bm{t}\left(\mathbf{x}\right)-\phi\left(\bm{\lambda}_{-i}\right)\right\} \mbox{d}\mathbf{x}
\]

\item Update global parameter $\bl_{s}\leftarrow\nabla\phi^{-1}\left(\be_{i}\right)$,
site parameter $\bl_{i}\leftarrow\bl_{s}-\sum_{j\neq i}\bl_{j}$
\end{enumerate}
\protect\caption{Classic EP in exponential family form\foreignlanguage{british}{\label{alg:Classic-EP}}}
\end{algorithm}

Classic EP (Alg. \ref{alg:Classic-EP}) loops over the sites sequentially. 

A parallel variant forms all the hybrids at once, looping several
times over the whole dataset (Alg. \ref{alg:Parallel-EP}). 

\begin{algorithm}
Loop until convergence
\begin{enumerate}
\item Process all hybrids: for $i$ in $1\ldots n$

\begin{enumerate}
\item Compute ``cavity'' parameters $\bl_{-i}\leftarrow\mathbf{\bl}_{s}-\bl_{i}$
\item Form hybrid distribution and compute moments 
\[
\mathbf{\be}_{i}=\frac{1}{Z_{i}}\int\bm{t}\left(\mathbf{x}\right)l_{i}\left(\mathbf{x}\right)\exp\left\{ \bm{\lambda}_{-i}^{t}\bm{t}\left(\mathbf{x}\right)-\phi\left(\bm{\lambda}_{-i}\right)\right\} \mbox{d}\mathbf{x}
\]

\item Compute local update $\bl_{i}\leftarrow\nabla\phi^{-1}\left(\be_{i}\right)-\sum_{j\neq i}\bl_{j}$
\end{enumerate}
\item Update global parameters $\bl_{s}\leftarrow\sum\bl_{i}$
\end{enumerate}
\protect\caption{Parallel EP\foreignlanguage{british}{\label{alg:Parallel-EP}}}
\end{algorithm}

\subsection*{Averaged EP}

We introduce a simpler variant of EP with a drastically reduced parameter
set: namely, we get rid of all site-specific parameters $\bl_{i}$
and keep only global parameters $\bl_{s}$. The resulting algorithm
is simpler to analyze. Our variant is straightforward, and follows
from setting $\bl_{i}=\frac{1}{n}\bl_{s}$ for all $i$, under the
assumption that the contributions from all sites are be similar.

Proceeding step-by-step from alg. \ref{alg:Parallel-EP} we begin
with the cavity parameter, which becomes $\bl_{c}=\bl_{s}-\frac{1}{n}\bl_{s}=\frac{n-1}{n}\bl_{s}$
independent of $i$. We use the cavity parameter to form hybrid distributions
just as before:

\[
h_{i}\left(\mathbf{x}\right)\propto l_{i}\left(\mathbf{x}\right)\exp\left\{ \frac{n-1}{n}\bl_{s}\bm{t}\left(\mathbf{x}\right)\right\} 
\]

The moments of the hybrids are again noted $\be_{i}$, and inserting
the local updates into the update for the global parameter we get
(recall that $\nabla\phi^{-1}$ transforms moment parameters into
natural parameters):

\begin{eqnarray}
\bl_{s}^{'} & = & \sum\left\{ \nabla\phi^{-1}\left(\be_{i}\right)-\bl_{c}\right\} =\sum\left\{ \nabla\phi{}^{-1}\left(\be_{i}\right)-\frac{n-1}{n}\bl_{s}\right\} \nonumber \\
 & = & \sum\nabla\phi{}^{-1}\left(\be_{i}\right)-\left(n-1\right)\bl_{s}\label{eq:simplified-update-rule}
\end{eqnarray}

It is interesting to examine the fixed points of this update rule,
which satisfy:

\begin{eqnarray*}
\mathbf{\bl}_{s}^{\star} & = & \sum\nabla\phi{}^{-1}\left(\be_{i}\vert\bl_{s}^{\star}\right)-\left(n-1\right)\bl_{s}^{\star}
\end{eqnarray*}

or equivalently:

\[
\mathbf{\bl}_{s}^{\star}=\frac{1}{n}\sum\nabla\phi^{-1}\left(\be_{i}\vert\bl_{s}^{\star}\right)
\]

where the hybrid moments $\be_{i}$ depend implicitly on $\bl_{s}$.
The following averaging rule shares the same fixed points:

\begin{equation}
\bl_{s}^{'}=\frac{1}{n}\sum\nabla\phi^{-1}\left(\be_{i}\vert\bl_{s}\right)\label{eq: averaged update}
\end{equation}

and that is the rule that gives averaged-EP (aEP) its name%
\footnote{Note that it corresponds to a slowed down version of the aEP update%
}.

The resulting method is given in Alg. \ref{alg:averaged-EP} but can
be summarized in a few words. To improve an exponential-family approximation
aEP begins by forming $n$ hybrids of the approximation and the true
posterior, it computes their moments, uses those to compute the new
site approximations and the corresponding natural parameters, and
sets the new natural parameters of the approximation to the sum of
the site approximations.

\begin{algorithm}
Loop until convergence
\begin{enumerate}
\item Compute ``cavity'' parameters $\bl_{c}\leftarrow\frac{n-1}{n}\mathbf{\bl}_{s}$
\item For $i$ in $1\ldots n$, form hybrid distribution and compute moments
\[
\mathbf{\be}_{i}\leftarrow\frac{1}{Z_{i}}\int\bm{t}\left(\mathbf{x}\right)l_{i}\left(\mathbf{x}\right)\exp\left\{ \bm{\lambda}_{c}^{t}\bm{t}\left(\mathbf{x}\right)-\phi\left(\bm{\lambda}_{c}\right)\right\} \mbox{d}\mathbf{x}
\]

\item Update global parameters $\bl_{s}^{'}\leftarrow\sum\nabla\phi{}^{-1}\left(\be_{i}\right)-\left(n-1\right)\bl_{s}$
\end{enumerate}
\protect\caption{averaged-EP\label{alg:averaged-EP}}
\end{algorithm}

\section{Asymptotic behavior of the EP and aEP updates \label{sec:Asymptotically-exact}}

In this section, we investigate the dynamics of the EP and aEP algorithms.
We first present a new key result on the asymptotic behavior of the
EP approximation of a site: we show that as the variance of the cavity
$q_{-i}$ converges to 0, the approximation converges to a simple
Taylor approximation of $\log\left(l_{i}\right)$. This asymptotic
behavior has several consequences but we present here the most important
one: in the limit where all cavity priors $q_{-i}$ have small variance,
the parallel EP and the aEP updates converge towards the updates of
Newton's algorithm. A corollary is that, for multimodal target distributions,
all modes which have sufficient curvature have an associated EP fixed
point and that, as a certain measure of mode peakedness goes to infinity,
the EP fixed point converges to the CGA at that mode. Finally, this
enables us to prove that EP is asymptotically exact in the large-data
limit (if the CGA also is).

\subsection{Assumptions}

Throughout this section, we work in the one-dimensional case since
it is the easiest to understand. All results are straightforward to
extend to the $p$-dimensional case (i.e., when the target distribution
is $p-$dimensional), the most significant difficulty being notation.
In the appendix, we give the proofs for the \emph{$p$}-dimensional
case.

We use two assumptions on the sites $l_{i}\left(x\right)$. 

Both our conditions concern the negative log-likelihood of the sites
$\phi_{i}\left(x\right)=-\log\left(l_{i}\left(x\right)\right)$. Our
first assumption is that the second log-derivative of the sites has
a bounded range: there exists $B$ such that:
\begin{equation}
\forall i,x\ \ \max\left(\phi_{i}^{''}\left(x\right)\right)-\min\left(\phi_{i}^{''}\left(x\right)\right)\leq B\label{eq: bounded curvature condition}
\end{equation}

Our second condition concerns bounding some higher log-derivatives
of the sites, which ensures that all sites are sufficiently regular
so that we can use Taylor expansions and bound the remainder terms.
Our assumption is simply that there exist bounds $K_{3}$ and $K_{4}$
which bound the third and fourth derivatives of all $\phi_{i}$ functions.
For $d\in\left\{ 3,4\right\} $:
\begin{equation}
\forall i,x\ \ \left|\phi_{i}^{\left(d\right)}\left(x\right)\right|\leq K_{d}\label{eq: second assumption: bounded log-derivatives}
\end{equation}

Both of those conditions are easy to check in practice. For example,
for a Generalized Linear Model, we would simply need to check the
derivatives of the link function and that the design matrix is bounded
and of full column rank. The one important case for which we cannot
apply our result concerns non-parametric models, and, more generally,
cases in which $p$ is not fixed but grows. This reflects a limitation
of our proof, rather than one of EP, which works just fine in such
cases (see Appendix for details).

An important thing to note is that we chose those two assumptions
because they give very simple expressions for the error of the asymptotic
expression, but the limit behavior we present can still be reached
even if they are broken. In the appendix, we show how weaker assumptions
(bounded $l_{i}\left(x\right)$ and local smoothness of $\phi_{i}\left(x\right)$)
are sufficient to obtain our results on the limit behavior with similar
asymptotic errors.

\subsection{Limit behavior of the EP update}

\subsubsection{Limit behavior of the site update}

The only complicated step in EP and aEP (especially in practical implementation)
is the site-approximation update during which we form the hybrid distribution,
compute its moments and then subtract the contribution of the cavity
to obtain the approximation of the site. We study here the limit behavior
of the site-approximation as the cavity becomes more and more precise.
The result we obtain is essential to the rest of this work, but not
entirely intuitive, so our explanation will be progressive and careful.

What we are interested in is the limit behavior of the site update,
as the precision of the cavity becomes large. The reason we focus
on the high-precision limit is that, when there are many sites (datapoints),
each individual one makes a small contribution compared to the rest.
The cavity represents the contribution of all the other sites, and
generally speaking the more sites there are the lower the variance
of the cavity (the higher the precision). In large-data settings,
the cavity prior tends to dominate the site's likelihood, meaning
that at the level of individual sites, \emph{the ``large data''
limit becomes a ``weak data'' limit. }

To study that limit, our first object of interest is naturally the
hybrid:

\[
h_{i}\left(x\right)\propto l_{i}\left(x\right)\exp\left(-\frac{\beta}{2}x^{2}+\left(\beta\mu_{0}\right)x\right)
\]

where we have parametrized the cavity precision as $\beta$, and the
cavity mean stays constant (at $\mu_{0}$ throughout) %
\footnote{In the notation of the previous section, the natural parameters are
$\bl=\left[\begin{array}{cc}
\beta & \beta\mu_{0}\end{array}\right]^{t}$, the precision and linear shift.%
}. As $\beta$ grows large, the cavity prior (the Gaussian part) outweighs
the likelihood, and the hybrid starts to resemble a Gaussian centered
at $\mu_{0}$ with variance $\beta^{-1}$. Indeed those are provably
the limits of the mean and variance of $h_{i}$ when $\beta\rightarrow\infty$. 

When $\beta$ is large, the hybrid is almost the same as the cavity,
and it is tempting to conclude that when $\beta$ is large no update
happens (the cavity prior outweighs the likelihood $l_{i}$, the site
becomes negligible). That line of reasoning, although tempting, is
misleading, as an examination of the case of a Gaussian site shows.
Suppose 
\[
l_{i}\left(x\right)=\exp\left(-\frac{\gamma}{2}x^{2}+\alpha x\right)
\]

then, regardless of how large $\beta$ is, it is straightforward to
show that the site's natural parameters are always $\beta_{i}=\gamma$
and $r_{i}=\alpha.$ In other words: even when the prior outweighs
the likelihood, the site always increases the overall precision by
an additive factor and contributes to the overall linear shift.

In the non-Gaussian case the site's natural parameters also have a
non-trivial limit. The exact form of that limit turns out to be very
interesting, as it shares a close relationship to Newton's method.
Specifically, we show that $r_{i}$ reflects the gradient of the log-likelihood
at the cavity mean and $\beta_{i}$ the Hessian. In other words, the
log of the site-approximation tends towards the Taylor expansion around
$\mu_{0}$ of the log-site $\phi_{i}\left(x\right)$. Fig. \ref{fig:Site-limits}
illustrates that behavior in a simple scenario, where:

\begin{equation}
h_{i}\left(x\right)=\left(\frac{1}{1+e^{-x}}\right)\exp\left(-\frac{\beta}{2}x^{2}\right)\label{eq:logit-hybrid}
\end{equation}

which corresponds to a logit likelihood and a cavity prior centered
at 0. Here $\phi_{i}=\log\left(1+e^{-x}\right)$, $\phi{}_{i}^{'}\left(0\right)=-\frac{1}{2}$,
and $\phi{}_{i}^{''}\left(0\right)=\frac{1}{4}$. 

\begin{figure}
\begin{centering}
\includegraphics[width=8cm]{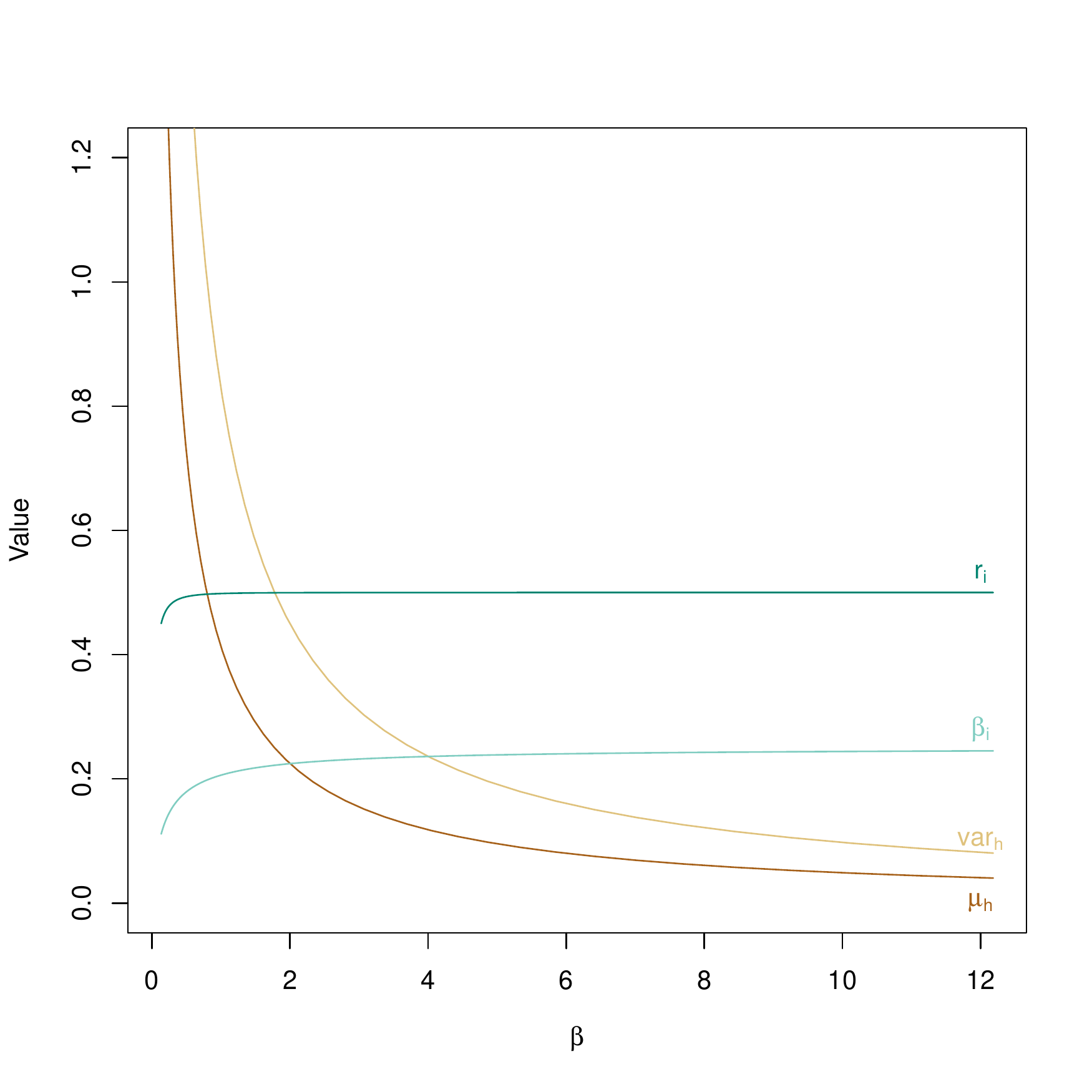}
\par\end{centering}

\protect\caption{Limit of site updates under increasing cavity precision. We use the
example given by eq. \eqref{eq:logit-hybrid}, where the site is a
logistic likelihood, the cavity prior has mean 0 and precision $\beta.$
The quantities shown are: the mean and variance of $h_{i}$ (labeled
$\mu_{h}$ and $var_{h})$, and the site parameters $r_{i}$ and $\beta_{i}$.
As expected, the mean tends to 0, and the variance tends to $\beta^{-1}$.
More surprisingly, the site parameters have non-trivial limits that
can be computed exactly (see main text). \label{fig:Site-limits} }
\end{figure}

We can now state our result formally. For simplicity, the case we
have just discussed had increasing precision and a fixed mean. The
following theorem is stated in a slightly more general case in which
the cavity mean is slightly offset from $\mu_{0}$, but tends to it
in as $\beta\rightarrow+\infty$. This more general case is important
for the corollaries we derive from this theorem.
\begin{thm}
Limit behavior of hybrid distributions\label{thm:Limit-behavior-of the site approximation}

Consider the hybrid distribution: $h_{i}(x)=l_{i}\left(x\right)\exp\left(-\frac{\beta}{2}x^{2}+\left(\beta\mu_{0}-\delta r\right)x\right)$.
In the limit that $\beta\rightarrow\infty$, the natural parameters
of its Gaussian approximation converge to:
\begin{eqnarray*}
\mbox{var}_{h_{i}}^{-1}E_{h_{i}} & \approx & \mbox{var}_{h_{i}}^{-1}\mu_{0}-\delta r+\phi_{i}^{'}\left(\mu_{0}\right)\\
\mbox{var}_{h_{i}}^{-1} & \approx & \beta+\phi_{i}^{''}\left(\mu_{0}\right)
\end{eqnarray*}

Thus, the natural parameters of the EP approximation ($r_{i}=var_{h_{i}}^{-1}E_{h_{i}}(x)-(\beta\mu_{0}-\delta r)$
and $\beta_{i}=var_{h_{i}}^{-1}-\beta$) of $l_{i}$ converge:
\begin{eqnarray*}
r_{i} & = & -\phi_{i}^{'}(\mu_{0})+\beta_{i}\mu_{0}+\mathcal{O}\left(\beta^{-1}+\left|\delta r+\phi_{i}^{'}(\mu_{0})\right|\beta^{-1}\right)\\
\beta_{i} & = & \phi_{i}^{''}(\mu_{0})+\mathcal{O}\left(\beta^{-1}+\left|\delta r+\phi_{i}^{'}(\mu_{0})\right|\beta^{-1}\right)
\end{eqnarray*}

\end{thm}
Note the important role of the $\delta r$ term: it causes the cavity-mean
to be slightly different from $\mu_{0}$, but it can accelerate convergence
when set precisely to $\delta r=-\phi_{i}^{'}\left(\mu_{0}\right)$
(see appendix).
\begin{proof}
We only give a sketch of the proof here, because it is too long and
involved.

Our proof can be understood as simply computing the asymptotic behavior
of $E_{h_{i}}\left(x\right)$ and $\text{var}_{h_{i}}\left(x\right)$.
The first order is easily found to be: $E_{h_{i}}\left(x\right)\approx\mu_{0}$
and $\text{var}_{h_{i}}\left(x\right)\approx\beta^{-1}$. However,
when we compute the new values for $r_{i}$ and $\beta_{i}$, the
subtraction of the cavity parameters effectively cancels that first
order term. In our proof, we thus go beyond the first order and compute
the next order which gives us the claimed bound.

In practice, we use two tricks that enable us to directly express
$\beta_{i}$ and $r_{i}$ as expected values under the hybrid $h_{i}$,
which saves us from actually computing the limit behavior of the mean
and variance. We then approximate these expected values using Taylor
expansions and get the claimed result. See the appendix for details.
\end{proof}

\subsubsection{Limit behavior of parallel-EP and aEP.}

Now that we have some handle on the behavior of site-updates, we can
start to study the behavior of the algorithm as a whole. A full step
of aEP or parallel EP is a combination of $n$ site updates, and that
is what we characterize next. We show that one step of parallel-EP
or of one step of aEP both converge towards the result of one step
of Newton's algorithm. It is also possible to use Theorem \ref{thm:Limit-behavior-of the site approximation}
to describe the limit behavior of sequential-EP or of an EP variant
which updates batches of sites sequentially, which tend to variants
of Newton's%
\footnote{For example, sequential EP would asymptote to a variant of sequential
gradient descent with a Hessian correction. See \citet{Opper1999}
for a more extensive discussion of the link between sequential EP
and sequential gradient descent.%
}. We choose to focus on parallel-EP because its limiting behavior
is classic Newton's.

Let's first present the limit behavior of aEP, which is easier to
visualize because it only has two parameters instead of $2n$-parameters
like EP. One interesting feature of the limit behavior of the site-update
is that the value of the cavity precision $\beta_{-i}$ does not influence
the limit behavior, which is only set by the cavity mean $\mu_{-i}$.
In the aEP algorithm, the cavity mean is always equal to the current
approximation mean. Thus, when we sum all $r_{i}$ and $\beta_{i}$
approximations, we find that the limit behavior of the aEP update
also has that feature: the approximation at the next step mostly depends
on the current mean of the approximation, and corresponds to a Newton's
update.

The limit behavior of EP is similar, but is a little more complicated
to state. This is due to two additional complications. The first complication
is that, in EP, each cavity mean $\mu_{-i}$ is slightly different.
This is where the $\delta r$ parameter from Theorem \ref{thm:Limit-behavior-of the site approximation}
comes into play: it enables us to see each cavity distribution instead
as almost centered at the same mean but slightly offset in a specific
direction. The second complication is that, whereas in aEP each cavity
distribution has the same precision, once again each cavity distribution
is different in EP. In the end, these complications hardly matter
for the limit behavior, but they do make it slightly harder to understand
how EP works.
\begin{thm}
Limit behavior of aEP and EP\label{thm:Limit-behavior-of aEP and EP}

Consider a current EP approximation $\left(r_{i,}\beta_{i}\right)_{i\in\left[1,n\right]}$
and the corresponding aEP approximation $\left(r=\sum r_{i},\beta=\sum\beta_{i}\right)$
whose current mean is $\mu_{0}=\frac{\sum r_{i}}{\sum\beta_{i}}$.
In the limit that all cavity-precisions $\beta_{-i}=\sum_{j\neq i}\beta_{j}$
tend to infinity (so that $\min\left(\beta_{-i}\right)$ is of same
order as $\beta$), the limit behavior of one step of aEP and of one
step of EP is identical to Newton's algorithm. 

For aEP, the global parameters at the next step are:
\begin{eqnarray*}
r_{aEP} & = & -\psi^{'}(\mu_{0})+\beta_{aEP}\mu_{0}+\mathcal{O}\left(n\beta^{-1}+\sum_{i}\left|\phi_{i}^{'}(\mu_{0})\right|\beta^{-1}\right)\\
\beta_{aEP} & = & \psi^{''}(\mu_{0})+\mathcal{O}\left(n\beta^{-1}+\sum_{i}\left|\phi_{i}^{'}(\mu_{0})\right|\beta^{-1}\right)
\end{eqnarray*}

For EP, the global parameters at the next step are:
\begin{eqnarray*}
\sum_{i}r_{i}^{new} & = & -\psi^{'}(\mu_{0})+\left(\sum_{i}\beta_{i}^{new}\right)\mu_{0}+\mathcal{O}\left(n\beta^{-1}+\sum_{i}\left|r_{i}-\beta_{i}\mu_{0}+\phi_{i}^{'}(\mu_{0})\right|\beta^{-1}\right)\\
\sum_{i}\beta_{i}^{new} & = & \psi^{''}(\mu_{0})+\mathcal{O}\left(n\beta^{-1}+\sum_{i}\left|r_{i}-\beta_{i}\mu_{0}+\phi_{i}^{'}(\mu_{0})\right|\beta^{-1}\right)
\end{eqnarray*}
\end{thm}
\begin{proof}
This result is simply obtained by summing the approximations offered
by Theorem \ref{thm:Limit-behavior-of the site approximation}.

For aEP, this is simple enough: all the cavity distributions are Gaussians
with precision $\frac{n-1}{n}\beta$ and with mean $\mu_{0}$. Straightforward
application of theorem \ref{thm:Limit-behavior-of the site approximation}
leads to the claimed result.

For EP, this is more complicated since every cavity distribution is
different. However, it is straightforward to check that the cavity
densities are: 
\[
g_{-i}\left(x\right)\propto\exp\left[-\left(\beta-\beta_{i}\right)\frac{x^{2}}{2}+\left(\left(\beta-\beta_{i}\right)\mu_{0}+\beta_{i}\mu_{0}-r_{i}\right)x\right]
\]
 We can then apply theorem \ref{thm:Limit-behavior-of the site approximation}
with cavity precision $\beta-\beta_{i}$ and offset $\delta r=r_{i}-\beta_{i}\mu_{0}$,
and recover the claimed result.
\end{proof}

\subsection{Where to find EP's fixed points}

In this section, we use the results above to find out more about the
location of fixed points of EP and aEP. We show that wherever the
posterior distribution has a strongly peaked mode, a fixed point of
EP or aEP lies in the vicinity. Our proof relies on an application
of Brouwer's fixed point theorem, and relies on finding \emph{stable
regions }of the parameter space, in a sense we need to make precise.

Since Newton's iterations are strongly contractive towards posterior
modes, and since our results tell us that the iterations of EP and
aEP are not far from those of Newton's, there is a good chance EP
and aEP do not stray too far from posterior modes either. Indeed,
we prove that there exist compact regions of the parameter space near
the CGA which are stable under the aEP or the EP updates: i.e., if
we start from inside of them, we stay inside. We can picture these
stable regions as boxes in parameter spaces inside of which aEP and
EP get stuck.

Unfortunately, our bounds are too weak to guarantee that the iterations
of aEP and EP converge in such regions. We know that they cannot exit
the box, but we cannot prove that they do not wander around forever
inside of it. However, there is a much more interesting consequence
of the existence of such stable regions: from the Brouwer fixed-point
theorem, we know that any compact stable region must contain at least
one fixed-point of the corresponding iteration, and so we have boxes
in parameter spaces that contain both a fixed point of Newton's and
a fixed point of aEP/EP. In order to apply this insight, we would
then want to find stable regions that are as small as possible in
order to give the tightest bounds on the position of that fixed-point.

In this section, we focus on identifying stable regions that are a
close neighborhood of the CGA at the mode of the target distribution,
and we compute the correct asymptotic scaling of the size of the stable
region. These results are sufficient to prove that aEP and EP are
both exact in the large-data limit. However, it would be an interesting
extension of the present work to also find maximal stable regions,
and to find ``unstable'' regions: regions of the parameter space
that the EP iteration is guaranteed to leave and which therefore cannot
hold a fixed-point.

We find that all modes of $p\left(x\right)$ have the potential to
have an associated stable region, and that the size of that stable
region depends on log-curvature at the mode: more peaked modes have
an associated region that is smaller than flatter modes. We use this
result in the next section to prove that aEP and EP fixed points converge
to the CGA at the mode in the large-data limit. 

We use aEP to outline this result. Let's assume that the starting
global approximation is in close proximity to the CGA at a mode $x^{\star}$
of $p\left(x\right)$:
\begin{eqnarray*}
\left|r_{aEP}-\beta_{aEP}x^{\star}\right| & \leq & n\Delta_{r}\\
\left|\beta_{aEP}-\psi^{''}\left(x^{\star}\right)\right| & \leq & n\Delta_{\beta}
\end{eqnarray*}
By applying Theorem \ref{thm:Limit-behavior-of the site approximation}
to a Gaussian approximation centered at $x^{\star}$, we find that
the new value of the aEP parameters are such that $r_{aEP}^{new}-\beta_{aEP}^{new}x^{\star}$
is small: 
\begin{equation}
r_{aEP}^{new}-\beta_{aEP}^{new}x^{\star}=\mathcal{O}\left(n\left[\psi^{''}\left(x^{\star}\right)-\Delta_{\beta}\right]^{-1}+\left[n\Delta_{r}+\sum_{i}\left|\phi_{i}^{'}(x^{\star})\right|\right]\left[\psi^{''}\left(x^{\star}\right)-\Delta_{\beta}\right]^{-1}\right)\label{eq:aEP-error-around-fixed-point}
\end{equation}

If the error is smaller than $n\Delta_{r}$, the initial region would
be stable in $r$. In order to find the limits of the stable region,
we simply need to find values $\Delta_{\beta}$ and $\Delta_{r}$
for which we can guarantee that the error is strictly smaller. Inspection
of eq. \eqref{eq:aEP-error-around-fixed-point} suggests immediately
that the curvature at the mode (represented by $\psi^{''}\left(x^{\star}\right)$)
plays a key role: the larger the curvature, the tighter the bound. 

For EP, the stable regions take the form:
\begin{eqnarray*}
\left|r_{i}+\phi_{i}^{'}\left(x^{\star}\right)-\beta_{i}x^{\star}\right| & \leq & \Delta_{r}\\
\left|\beta_{i}-\phi_{i}^{''}\left(x^{\star}\right)\right| & \leq & \Delta_{\beta}
\end{eqnarray*}
which ensures the global approximation is inside stable regions with
the same form as those for aEP.

$\Delta_{r}$ and $\Delta_{\beta}$ are small if the log-curvature
at the mode is sufficiently high, as summarized by the following theorem:
\begin{thm}
Convergence of fixed points of EP and aEP\label{thm:Convergence-of-fixed-points of EP and aEP}

There exists an EP and an aEP fixed-point close to the CGA of $p\left(x\right)$
at $x^{\star}$ if $\phi_{i}^{''}\left(x^{\star}\right)$ is sufficiently
large. More precisely, if:
\begin{eqnarray*}
\delta_{aEP} & = & \max\left(K_{3},K_{4}\right)\sum\left|\phi_{i}^{'}\left(x^{\star}\right)\right|\left[\psi^{''}\left(x^{\star}\right)\right]^{-1}\\
\delta & = & n\max\left(K_{3},K_{4}\right)\left[\psi^{''}\left(x^{\star}\right)\right]^{-1}
\end{eqnarray*}
 are $\mathcal{O}\left(1\right)$ quantities and $\psi^{''}\left(x^{\star}\right)$
is large, then the limit of the stable regions on the global approximation,
$n\Delta_{r}$ and $n\Delta_{\beta}$, scale as $\mathcal{O}\left(\delta_{aEP}+\delta\right)$
for aEP and as $\mathcal{O\left(\delta\right)}$ for EP.\end{thm}
\begin{proof}
We only sketch the proof of this theorem which we detail in the appendix.
We focus on the simpler aEP case but the reasoning is identical for
EP.

The key idea is the following: if we perform a first order perturbation
in $r_{aEP}$ or in $\beta_{aEP}$ while $\beta_{aEP}$ is large,
then this perturbation has a negligible effect on the limit behavior.
Thus, the error still scales (almost) as if we were starting from
the CGA at $x^{\star}$: $\beta_{aEP}=\psi^{''}\left(x^{\star}\right)$
and $r_{aEP}=\beta_{aEP}x^{\star}$:
\begin{eqnarray*}
r_{aEP}^{new} & = & \beta_{aEP}^{new}x^{\star}+\mathcal{O}\left(n\max\left(K_{3},K_{4}\right)\left[\psi^{''}\left(x^{\star}\right)\right]^{-1}\right)\\
\beta_{aEP}^{new} & = & \psi^{''}\left(x^{\star}\right)+\mathcal{O}\left(\max\left(K_{3},K_{4}\right)\left[n+\sum_{i=1}^{n}\left|\phi_{i}^{'}\left(x^{\star}\right)\right|\right]\left[\psi^{''}\left(x^{\star}\right)\right]^{-1}\right)
\end{eqnarray*}
from which we get the claimed limit behavior.\end{proof}
\begin{rem*}
It might seem strange to refer to $\delta$ and $\delta_{aEP}$ as
``order 1'' quantities. This holds in the large-data regime as we
show in the next section, but an easier example to visualize is the
following. Consider the EP approximation of a fixed-probability distribution
raised to power $\lambda$: $\left[p\left(x\right)\right]^{\lambda}=\prod_{i}\left[l_{i}\left(x\right)\right]^{\lambda}$.
For that example, as $\lambda\rightarrow\infty$, $\lambda\psi^{''}\left(x^{\star}\right)\rightarrow\infty$:
the log-derivative at the mode grows linearly. However, the third
and fourth log-derivatives also grow linearly so that $\delta=n\max\left(K_{3},K_{4}\right)\left[\psi^{''}\left(x^{\star}\right)\right]^{-1}$
does not depend on $\lambda$: it is indeed of order 1. $\delta_{aEP}$
is similarly found to be of order 1.
\end{rem*}
This theorem shows two interesting features of the behavior of EP.
First of all, we see that if $p\left(x\right)$ is a multimodal distribution
where multiple modes are sufficiently peaked and they are sufficiently
separated, then EP and aEP both have multiple fixed points, and those
fixed points do not give a global account of $p\left(x\right)$ but
only fit the local shape of $p\left(x\right)$ around ``their''
mode. This is quite contrary to the common view on EP which holds
that since EP's stated target is to find an approximation of the minimizer
of $KL\left(p,q\right)$, it gives global approximations of the target
distribution. We discuss this point further in section \ref{sub:EP-behavior-on multimodal distributions}.

\subsection{Large-data limit behavior}

So far, all of our results have been deterministic: assuming some
fixed target distribution $p\left(x\right)$, we have bounded the
distance between the result of the aEP and EP updates and the result
of the NT updates. We then used those results to derive a deterministic
result on the possible positions of the aEP and EP fixed points. We
have discussed asymptotic results in those sections in terms of either
asymptotes of the parameter space (large precision $\beta$) or of
properties of the target distribution (large log-curvature at a mode
$\psi^{''}\left(x^{\star}\right)$).

In this section, we adopt a different point of view: we seek a large-data
limit result. In other words, we assume that some random process is
generating the sites $l_{i}\left(x\right)$ and we consider what happens
as more and more sites are generated. In real applications, this would
correspond to accumulating data of some kind and computing the posterior
of the unknown $x$ under some (most likely miss-specified) generative
model. We abstract all those complications away and simply treat the
$l_{i}$ functions (or, equivalently, the $\phi_{i}$ functions) as
function-valued random variables.

Throughout this section, the number of sites, $n$, is variable. We
note $p_{n}\left(x\right)$ the random variable of the posterior distribution
constructed from the $n$ first sites $l_{i}\left(x\right)$. We note
$x_{0}$ the minimum of $x\rightarrow E\left(\phi_{i}\left(x\right)\right)$:
$x_{0}$ can be thought of as the ``true'' value that we seek to
recover. We note $x_{n}^{\star}$ the mode of $p_{n}\left(x\right)$
closest to $x_{0}$ and $q_{n}\left(x\right)$ the CGA of $p_{n}\left(x\right)$
at $x_{n}^{\star}$.

In order for EP to have good behavior, we require the process generating
the $l_{i}$ to obey our assumptions on the $l_{i}$ (eqs. \eqref{eq: bounded curvature condition}
and \eqref{eq: second assumption: bounded log-derivatives}) and to
obey two additional conditions. The first condition is that the distribution
of the log-sites $\phi_{i}\left(x\right)$ is non-degenerate so that
a number of variances are finite and we can apply the law of large
numbers (see appendix for details). This is a mild condition and,
in the rare occasion where it does not apply, we could even weaken
it. This condition ensures that the log-posterior $\sum\phi_{i}\left(x\right)$
is Locally Asymptotically Normal (LAN, \citet{KleijnVanDerVaart:BernsteinVonMisesUnderMisspec}).
Under this LAN behavior, we prove that every global approximation
in the aEP and EP stable regions around $x_{n}^{\star}$ converge
in KL divergence and in total-variation towards the CGA at $x_{n}^{\star}$:
$q_{n}\left(x\right)$.

Furthermore, if the process generating the $l_{i}$ produces a concentration
of the mass of the posterior around $x_{0}$, then, combined with
the LAN behavior of the posterior, we have enough to guarantee that
the posterior $p_{n}\left(x\right)$ converges towards its CGA $q_{n}\left(x\right)$
in total-variation. This result is the last we need to prove that,
in the large-data limit, aEP and EP are exact in the following sense:
there is a large neighborhood of aEP and EP approximations that surrounds
at least one fixed-point and where the aEP and EP iterations are ``stuck'',
such that all approximations in the neighborhood are asymptotically
exact.

The technical condition we require for concentration of mass is the
following: for all $\epsilon>0$, the integrals $\int p_{n}\left(x\right)1\left(\left|x-x_{0}\right|\leq\epsilon\right)dx$,
which are random variables whose distribution is dictated by the distribution
of the $l_{i}\left(x\right)$, need to converge in probability to
1. In other words, for every $\epsilon>0$, the posterior is guaranteed
to concentrate inside the $\epsilon$-ball centered around $x_{0}$.
This should be thought of as an identifiability condition: it requires
the posterior to concentrate around the ``true'' parameter value
$x_{0}$.
\begin{thm}
aEP and EP are exact in the large-data limit.\label{thm:aEP-and-EP are exact}

Under our assumptions on the sites and the site-generating process,
all Gaussian distributions in the stable region of th. \ref{thm:Limit-behavior-of aEP and EP}
converge in total-variation to the CGA $q_{n}$ with probability 1.
The convergence rate is $\mathcal{O}\left(n^{-1/2}\right)$.

Under a further identifiability assumption, all Gaussian distributions
in the stable region converge in total-variation to $p_{n}$ with
probability 1. The convergence rate is $\mathcal{O}\left(n^{-1/2}\right)$.\end{thm}
\begin{proof}
Here is a sketch of the proof. First, define $I_{0}=E\left(\phi_{i}^{''}\left(x_{0}\right)\right)$
the Fisher information of our likelihood-generating process.

By a law of large numbers argument, $\sum_{i=1}^{n}\phi_{i}^{''}\left(x_{0}\right)\approx nI_{0}$:
this quantity grows linearly. In the meantime, the gradient at $x_{0}$
is small: $\sum_{i=1}^{n}\phi_{i}^{'}\left(x_{0}\right)\approx\sqrt{n}\sqrt{\text{var}\left(\phi_{i}^{'}\left(x_{0}\right)\right)}$.
Combining this with the third derivative bound and a simple Taylor
expansion of $\sum_{i=1}^{n}\phi_{i}^{'}\left(x_{0}\right)$ proves
that there must be a mode of $p_{n}$: $x_{n}^{\star}$, in close
proximity to $x_{0}$. The distance between the two scales as: $x_{n}^{\star}-x_{0}=\mathcal{O}\left(1/\sqrt{n}\right)$.

Since the log-curvature at $x_{0}$ grows linearly and $x_{0}-x_{n}^{\star}=\mathcal{O}\left(1/\sqrt{n}\right)$,
the log-curvature at $x_{n}^{\star}$ also grows linearly:
\[
\sum_{i=1}^{n}\phi_{i}^{''}\left(x_{n}^{\star}\right)\approx nI_{0}
\]

Similarly, $\sum_{i=1}^{n}\left|\phi_{i}^{'}\left(x_{n}^{\star}\right)\right|$
grows linearly and $n\max\left(K_{3},K_{4}\right)$ trivially grows
linearly with $n$. Thus, the conditions of th. \ref{thm:Limit-behavior-of aEP and EP}
are checked: there exists a stable region near the CGA at $x_{n}^{\star}$
which holds at least one fixed-point for aEP and EP.

To prove the total-variation convergence, we actually prove a KL-divergence
convergence. The bounds on $\Delta_{r}$ and $\Delta_{\beta}$ translate
into a $\mathcal{O}\left(n^{-1}\right)$ KL-divergence bound. From
Pinsker's inequality, this translates into a $n^{-1/2}$ total-variation
bound.

The proof then concludes by proving a $n^{-1/2}$ convergence of the
CGA towards $p_{n}$ which is a simple application of a Bernstein-von
Mises theorem for miss-specified models by \citet{KleijnVanDerVaart:BernsteinVonMisesUnderMisspec}.
\end{proof}
As a corollary from this theorem, it follows that EP is asymptotically
exact in all models which respect our hypothesis on the $l_{i}$ and
the $l_{i}$-generating process and which are identifiable. This includes
an extremely large class of models since our conditions are fairly
mild and since identifiability is a key requirement for the Bayesian
method to be useful. For example, our result can be applied to both
probit and logistic regression in finite-dimensions, as long as the
feature vectors are bounded (in order for our hypotheses on the $l_{i}$
to be verified), and spread uniformly enough for the Fisher information
matrix to be strictly positive (see Appendix). Three notes must be
made on that theorem. 

First of all, it shows that EP and aEP are exact in that there exists
a fixed-point which converges in total-variation to the true posterior.
It does not guarantee that asymptotically all fixed points converge.
In particular, should the expected value of $\phi_{i}\left(x\right)$
have several modes, we can guarantee that there also exists aEP and
EP fixed points which are terrible asymptotic approximations of $p_{n}\left(x\right)$:
the stable regions associated with the local minima converge to the
CGA at a mode with negligible asymptotic contribution to the mass
of $p_{n}$.

Second, it could seem from this result that EP and aEP approximations
are asymptotically worse than the CGA, because the total-variation
distance between the EP fixed-point and the target decreases slower
than for the CGA. This does not reflect a limitation of EP but it
is a feature of our proof: we have proved that EP is good because
it converges to the CGA, and only a direct proof would be able to
prove the superiority of EP. We expect that EP should give better
asymptotic approximations than the CGA from empirical tests of both
methods, but the result we present here is too weak to prove this
conjecture. We have recently made some progress on such a direct proof,
but under more restrictive assumptions than the ones presented here
\citep{DehaeneBarthelmeNips2015}.

Finally, it is interesting to come back to th. \ref{thm:Limit-behavior-of aEP and EP}
which shows that, in the large cavity precision limit, aEP and parallel
EP converge to NT. It is also possible to qualitatively discuss th.
\ref{thm:Limit-behavior-of aEP and EP} in terms of a large-data limit
instead. In order to do so, we assume that, as the number of data-points
grows, the typical value for the cavity precisions $\beta_{-i}$ grows
linearly with $n$. This is certainly true in the stable region around
$x_{0}$ as we have just shown in th. \ref{thm:aEP-and-EP are exact}.
If we have that $\min\left(\beta_{-i}\right)\propto n$, then the
errors in th. \ref{thm:Limit-behavior-of aEP and EP} are of order
1. These order 1 errors are negligible in the large-data limit, in
exactly the same way as for th. \ref{thm:aEP-and-EP are exact}. Thus
our result that aEP and parallel EP are almost Newton in the high
precision limit qualitatively apply in the large-data limit as long
as the ``typical'' cavity precisions grow linearly with $n$.

\section{Consequences of the quasi-Newton behavior of EP\label{sec:Consequences-following-from the limit behavior}}

In the previous section, we gave a proof that, in the limit of large-data,
EP behaves like a Newton search for the mode of the target distribution.
In this section, we highlight how this result can inform our intuition
about how EP behaves, and some potentially interesting avenues of
research it opens.

\subsection{Instability of the EP iteration\label{sub:Instability}}

Newton's algorithm (NT) is a good tool in finding a mode of a target
distribution as it has fast convergence if it is initialized properly
(i.e.: close enough to the mode), but it can often fail to converge
globally. For example, applying NT to $f(x)=\exp(-|x|^{4/3})$ always
results in a divergent sequence that oscillates wildly around the
fixed point at $x=0$. More generally, Newton is unstable when the
log-curvature $\psi^{''}$ is small because that makes the Newton
step $\left[\psi^{''}\right]^{-1}\psi^{'}$ too big.

In order to fix this problem, it is necessary to introduce a ``slowed-down''
version of the iteration:

\[
\mu_{t+1}=\mu_{t}-\gamma_{t}\frac{\psi^{'}(\mu_{t})}{\psi^{''}(\mu_{t})}
\]

where $0\leq\gamma_{t}\leq1$ is chosen carefully to ensure convergence.
As the NT algorithm is part of the class of Generalized Gradient Descent
algorithms, one solution is to choose values $\gamma_{t}$ that respect
the Wolfe conditions (see \citealp{Boyd:2004:CO:993483}, for a convergence
analysis of Newton's method).

Since EP behaves like NT in the large-data limit, we can intuit that
even for small $n$, EP might have a qualitatively similar behavior.
In particular, EP iterations might oscillate around their fixed-point
just like NT does. We give here a simple example of this behavior
with sites that are extremely regular and which seem harmless at a
glance.

In our example, we applied a parallel version of the EP algorithm
to the following situation:
\begin{itemize}
\item five ``double-logistic'' sites: $\forall i\in\{1,\ldots,5\},\ l_{i}(x)=\left(1+\exp(5x)\right)^{-1}\left(1+\exp(-5x)\right)^{-1}$,
so called because they are the product of two logistic functions.
If plotted, these appear to be Gaussian at a glance, but with the
important difference that they only have exponential decay in their
tails. In these exponential tails, the log-curvature $\phi_{i}^{''}\left(x\right)$
is very small.
\item one Gaussian site representing a prior: $l_{0}(x)=\exp(-x^{2}/2)$
\end{itemize}
In this example, there is an EP fixed-point that provides a good approximation
of the target distribution. However, we also found that the EP iteration
is unstable if it is initialized too far away from the fixed-point
distribution. EP iterations initialized too far away converge to a
limit cycle oscillating between two approximations that are completely
wrong. Figure \ref{fig:Convergence-of-aEP} shows the basins of attraction
of the stable equilibrium and the limit cycle, and one example trajectory
for each.

Our results on the limit behavior of the EP iteration can thus inform
our understanding of why the EP iteration sometimes has problems with
convergence: it could be that the EP iteration overshoots when it
is operating in a zone where most $\phi_{i}^{''}\left(x\right)$ are
small while most $\phi_{i}^{'}\left(x\right)$ are big, exactly like
NT would. A possible solution to this could be to complement EP with
an adaptive step algorithm. This algorithm would need to detect overshoots
or potential overshoots, and prevent them. Finding such an adaptive
step algorithm would represent major progress in EP methods.

\begin{figure}
\begin{centering}
\includegraphics[width=10cm]{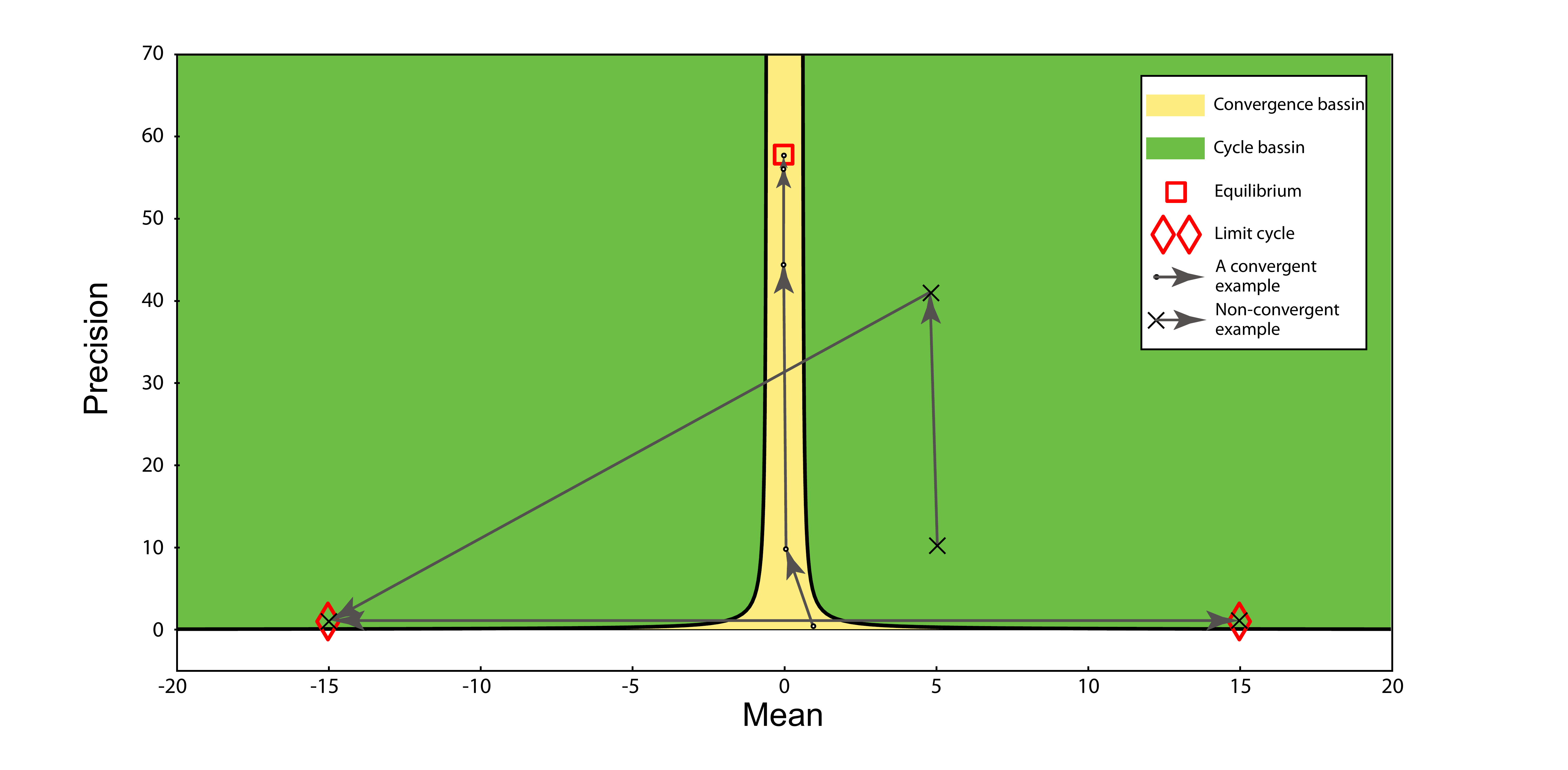}
\par\end{centering}

\protect\caption{Convergence of aEP on double-logistic sites. We evaluated the stability
of the averaged-EP algorithm on five double-logistic sites and one
Gaussian prior site. The aEP iteration either converged to a fixed
point (red square) or to an oscillation (between the two red diamonds).
The yellow background corresponds to the basin of attraction of the
fixed point and the green background to that of the limit cycle. The
first four steps of the aEP iteration are presented for two initial
points, one in each basin. Note how fast the convergence to both attractors
is: in four steps, the iteration reaches a very close neighborhood
of the attractors.\label{fig:Convergence-of-aEP}}
\end{figure}

\subsection{Behavior of EP on multimodal distributions\label{sub:EP-behavior-on multimodal distributions}}

Let's now investigate how our results shed new light on the behavior
of EP on a multimodal target distribution $p\left(x\right)$.

EP has been presented from the start as a rough approximation to the
minimizer of the ``forward'' KL divergence: $KL\left(p||q\right)$
that uses local (i.e., site-specific) approximations of the KL divergence.
This leads to the intuition that, when applied on a multimodal target,
the EP approximation would fit all modes, or maybe most modes, since
this is the behavior of the KL approximation.

With our method of bounding fixed points in neighborhoods of the CGA
at the various modes of $p\left(x\right)$, we can now see that this
intuition is flawed. Indeed, our result shows that all modes that
are:
\begin{itemize}
\item sufficiently peaked, so that the stable region is small
\item sufficiently isolated, so that their stable region does not overlap
with that of the other modes
\end{itemize}
have at least one associated fixed-point. This fixed-point corresponds
to the EP approximation fitting only this mode and not the rest of
the probability distribution. Thus, it can happen that EP approximations
give only a partial account of the target distribution.

However, it's also false to believe that EP always gets captured and
never provides a global approximation of a multimodal target. Indeed,
when there is only one site (or more generally, when there is only
a few sites), EP does give a global account of the distribution since
EP with one site exactly recovers the KL approximation of $p\left(x\right)$.

Surprisingly, both types of fixed points can co-exist. Figure \ref{fig:multimodal}
shows an example involving Gaussian mixtures, the prototypical example
of multimodal problems. Here the data $y_{1}\ldots y_{n}$ are supposed
IID, with 
\[
p(y_{i}|\mathbf{x})=\frac{1}{2}\left(\N\left(y_{i};x_{1},1\right)+\N\left(y_{i};x_{2},1\right)\right)
\]
The parameters correspond to component means in the Gaussian mixture
and are evidently interchangeable, so that the likelihood surface
is in general bimodal. We ran EP in this example with $n=20$, $\mathbf{x}=\left(0,-2.5\right)$
and a unit Gaussian prior on $\mathbf{x}.$ Different initializations
lead to different fixed points: we found three, two corresponding
to local approximations (as predicted by theory) and a global one,
the latter far from any mode. Interestingly, the local approximations
are locally ``exact'', meaning that under the identifiability constraints
$x_{2}>x_{1},$ the moments of the corresponding EP approximation
are exact. The mean of the global approximation matches the exact
global mean, although the covariance is a bit under-estimated. 

A simple take-home message from our work should thus be this one:
do not expect EP to fit all modes of a target distribution, but do
not automatically assume that it will fit a single mode either.

\begin{figure}
\begin{centering}
\includegraphics[width=8cm]{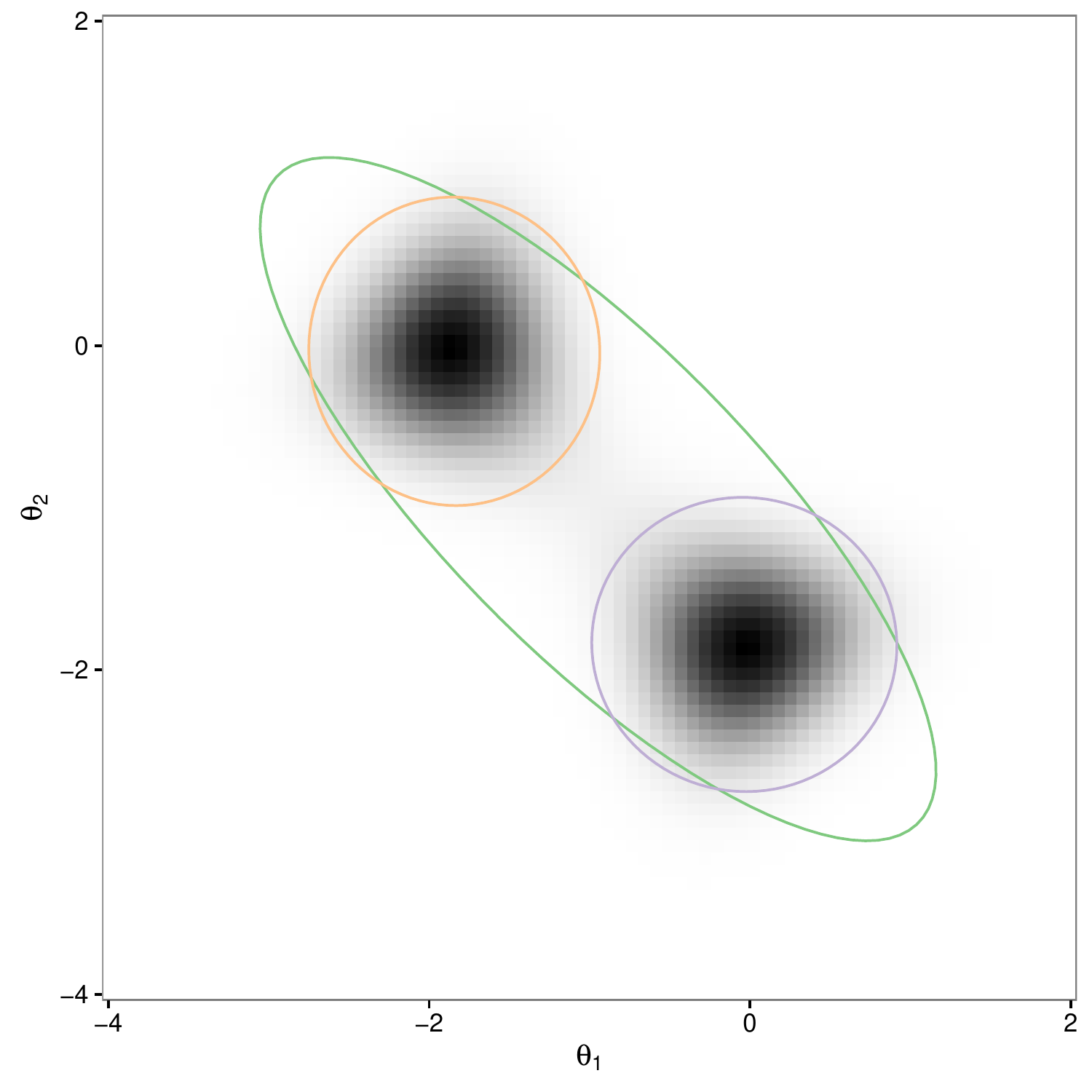}
\par\end{centering}

\protect\caption{Behavior of EP on multimodal target distributions. The grey density
represents the target, and the ellipses different EP approximations
(summarised as 95\% confidence regions). In this example EP has three
possible fixed points, one corresponding to a global approximation
of the target and the other two to local approximations. Which fixed
point is reached depends on the initialisation. See text for details.
\label{fig:multimodal}}
\end{figure}

\section{Conclusion\label{sec:Conclusion}}

EP is an algorithm whose theoretical analysis lags far behind its
empirical success. We describe in this manuscript a number of results
that narrow the gap between theory and empirics, and we hope that
they will provide a useful basis for future work.

In this article, we propose a simpler version of EP which we call
averaged-EP or aEP. aEP could be interesting as an empirical algorithm
(see Appendix and \citet{NIPS2015_5760} who introduce a close variant
of aEP called stochastic EP). However, our main focus is on using
it as a theoretical tool for studying the asymptotics of EP, since
its reduced parameter space makes the results simpler to understand.
We derive analytical results on aEP and EP in several limits: in the
limit of large cavity precisions, and in the classical large-data
limit. We prov that both methods converge to a Newton's search for
a mode of the target distribution. We then show that both are asymptotically
exact in that there exists a fixed-point which converges towards the
target.

Our theoretical results open several avenues of research into gaining
a better understanding of EP. First of all, while we shed some light
into the behavior of the EP iterations by providing a qualitative
link to Newton's method, we still do not know how to build a variant
of EP which is guaranteed to converge. This is a key avenue of research
since the only way we know to guarantee convergence to an EP fixed-point,
the Expectation-Consistent algorithm \citep{OpperWinther:ExpectationConsistentApproxInf},
converges much more slowly. An algorithm that always converges while
staying as fast as EP would represent a major step forward. The parallel
with NT opens the interesting idea of designing a line-search extension
of EP. 

Another limit of our result is the coarseness of our bounds: while
we show that EP is asymptotically exact, we do not show that it improves
on the Canonical Gaussian Approximation when there is ample empirical
evidence that it does. Future theoretical work on EP should aim at
showing how and when EP does dominate the CGA (see \citet{DehaeneBarthelmeNips2015}
for one such investigation, though crippled by unrealistic assumptions
on the model).

A final interesting extension of this work concerns the non-parametric
case, and, more generally, EP approximations of high-dimensional posteriors.
Indeed, we believe that our results are sub-optimal in bounding how
the error scales in high-dimensional cases, which is why we cannot
apply our results to the non-parametric case for which $p=n$. A careful
extension to show that EP behaves correctly in those cases would prove
another huge step forward in providing a good theoretical basis for
EP.

\section*{Acknowledgments}

We thank Alex Pouget for his support, and Hugo Duminil for helpful
insight on the math. We also thank M�lisande Albert, Nicolas Chopin,
Gina Gr�nhage, and James Ridgway for their comments on the manuscript.
Finally, we thank Judith Rousseau for her help on Bernstein-von Mises
theorems.

\bibliographystyle{apalike}
\bibliography{ref}

\global\long\def\bb{\bm{\beta}}
\global\long\def\bt{\bm{\theta}}
\global\long\def\bmu{\bm{\mu}}
\global\long\def\bl{\bm{\lambda}}
\global\long\def\be{\bm{\eta}}
\global\long\def\E{\mathbb{E}}
\global\long\def\by{\bm{y}}
\global\long\def\bys{\bm{y^{\star}}}
\global\long\def\bS{\bm{\Sigma}}
\global\long\def\N{\mathcal{N}}
\global\long\def\I{\mathds{1}}
\global\long\def\A{\mathcal{A}}
\global\long\def\Ml{\mathcal{M}_{l}}
\global\long\def\Q{\mathcal{Q}}
\global\long\def\cov{\text{Cov}}
\global\long\def\var{\mbox{var}}
\global\long\def\O{\mathcal{O}}

\section*{Supplementary information}

The following two sections hold all the supplementary information
of this article.

\section{Proofs}

In this section, we will give detailed proofs of all the results we
have presented in the main text.

We will prove, in order:
\begin{enumerate}
\item the limit behavior of the EP update in one-dimension
\item the limit behavior of the EP update in high-dimensions
\item the limit behavior under weaker assumptions
\item the exactness of aEP and EP in the large-data limit
\end{enumerate}

\subsection{Assumptions}

We will prove our results in the one-dimensional case and in the n-dimensional
case. Let's first recall our assumptions on the likelihoods in the
one-dimensional case. We will explain in section \ref{sub:Limit-behavior-in high-dimensions}
how to modify these assumptions in the high-dimensional case. In section
\ref{sub:Limit-behavior-under weaker assumptions}, we show that these
assumptions can be weakened considerably, though the expression for
the errors is much harder to state.

Let $l_{i}\left(x\right)=\exp\left(-\phi_{i}\left(x\right)\right)$
be the sites, and$\phi_{i}\left(x\right)$ be the negative log of
each site.

Our first assumption will be that the second log-derivatives of the
sites span a finite range:

\begin{equation}
\max\left(\phi_{i}^{''}\left(x\right)\right)-\min\left(\phi_{i}^{''}\left(x\right)\right)\leq B\label{eq: assumption 1, bound the range of the log-curvature}
\end{equation}

and second, that some of the higher log-derivatives are bounded. There
exists constants $K_{d}$ for $d\in\left[3,4\right]$ such that:

\begin{equation}
\left|\phi_{i}^{\left(d\right)}\left(x\right)\right|\leq K_{d}\label{eq: assumption 2, bound the 3 and 4 derivatives}
\end{equation}

\subsection{Limit behavior of the EP update}

In this section, we will prove the following theorem.
\begin{thm}
Limit behavior of the site-approximation\label{thm:Limit-behavior-of the site approximation-1}

Consider the hybrid distribution: $h_{i}(x)=l_{i}\left(x\right)\exp\left(-\frac{\beta}{2}x^{2}+(\beta\mu_{0}-\delta r)x\right)$.
In the limit that $\beta\rightarrow\infty$, the hybrid mean $E_{h_{i}}\left(x\right)$
and the natural parameters of the EP approximation ($r_{i}=var_{h_{i}}^{-1}E_{h_{i}}(x)-(\beta\mu_{0}-\delta r_{i})$
and $\beta_{i}=var_{h_{i}}^{-1}-\beta$) of $l_{i}$ converge. Defining
$K_{M}=\max\left(K_{3,}K_{4}\right)$ and $\Delta_{r}=\left(\phi_{i}^{'}\left(\mu_{0}\right)+\delta r\right)$,
the limits are:
\begin{eqnarray*}
E_{h_{i}}\left(x\right)=\mu_{h_{i}} & = & \mu_{0}+\O\left(\Delta_{r}\beta^{-1}+K_{M}\beta^{-2}\right)\\
r_{i} & = & \beta_{i}\mu_{h_{i}}-E_{h_{i}}\left(\phi_{i}^{'}\left(x\right)\right)\\
 & = & -\phi_{i}^{'}(\mu_{0})+\beta_{i}\mu_{0}+\O\left(K_{M}\beta^{-1}+K_{M}\beta^{-1}\left[\mu_{h_{i}}-\mu_{0}\right]+K_{M}\left[\mu_{h_{i}}-\mu_{0}\right]^{2}\right)\\
 & = & -\phi_{i}^{'}(\mu_{0})+\beta_{i}\mu_{0}+\O\left(K_{M}\beta^{-1}+\beta^{-2}\right)\\
\beta_{i} & \approx & E_{h_{i}}\left(\phi_{i}^{''}\left(x\right)\right)\\
 & = & \phi_{i}^{''}(\mu_{0})+\O\left(K_{M}\beta^{-1}+K_{M}\left[\mu_{h_{i}}-\mu_{0}\right]\right)\\
 & = & \phi_{i}^{''}(\mu_{0})+\O\left(K_{M}\beta^{-1}+K_{M}\Delta_{r}\beta^{-1}\right)
\end{eqnarray*}
\end{thm}
\begin{rem*}
Note the key role of the $\delta r$ parameter. It causes the mean
of the cavity distribution to be offset from $\mu_{0}$, but it can
make the mean of the hybrid, $\mu_{h_{i}}$, closer to $\mu_{0}$.
Indeed, if $\delta r$ is such that $\Delta_{r}=0$, then we gain
an order of magnitude in the limit behavior of $\mu_{h_{i}}$ and
the errors in the limit behavior of both $r_{i}$ and $\beta_{i}$
are smaller.\end{rem*}
\begin{proof}
Let's first sketch a global overview of how the proof works. Intuitively,
what is going on is that we are going beyond the first order approximations
of the mean and variance of $h_{i}$:
\begin{eqnarray*}
E_{h_{i}}\left(x\right) & \approx & \mu_{0}\\
\var_{h_{i}}\left(x\right) & \approx & \beta^{-1}
\end{eqnarray*}

and computing the next order of their limit behavior. If we try to
follow that path directly, however, we obtain bounds that are a bit
ugly and not very tight. A better proof path is slightly clever and
sophisticated and consists in finding ``tricks'' ways toof bounding
$\beta_{i}$ and $r_{i}$ directly. This is where the Brascamp-Lieb
inequality comes in play: it provides one-half of the $\beta_{i}$
bound.

In practice, our proof can be decomposed into five steps:
\begin{itemize}
\item Upper-bound $\var_{h_{i}}\left(x\right)$ in a coarse way
\item Upper-bound and lower-bound $\var_{h_{i}}\left(x\right)$ in a fine
way
\item Prove a coarse bound on the $E_{h_{i}}\left(x\right)-\mu_{0}=\mu_{h_{i}}-\mu_{0}$
from the coarse bound on the variance
\item Use the bound on $\mu_{h_{i}}=E_{h_{i}}\left(x\right)$ to improve
the bound on $\var_{h_{i}}\left(x\right)$ to its final state which
provides us with the bound on $\beta_{i}$
\item Compute the limit behavior of $r_{i}$ from the coarse limit behavior
of $\var_{h_{i}}\left(x\right)$ and $\mu_{h_{i}}$
\end{itemize}

\paragraph{Limit behavior of $\protect\var_{h_{i}}\left(x\right)$ and of $\beta_{i}=\protect\var_{h_{i}}\left(x\right)^{-1}-\beta$}

First, we will deal with the variance of the hybrid. We will use the
Brascamp-Lieb result and a Cramer-Rao like bound to derive the final
bounds on $\beta_{i}$. The Brascamp-Lieb result will also give a
coarse bound on $\var_{h_{i}}$ which we will use in the other sections.

Let's start by upper-bounding the variance with the Brascamp-Lieb
bound. This will also give us the coarse bound we need on $\var_{h_{i}}\left(x\right)$.

Consider the value of $\phi_{i}^{''}$ at $\mu_{0}$. It gives us
a universal lower-bound on $\phi_{i}^{''}\left(x\right)$ (from assumption
\ref{eq: assumption 1, bound the range of the log-curvature}):
\[
\phi_{i}^{''}\left(x\right)\geq\phi_{i}^{''}\left(\mu_{0}\right)-B
\]
Thus, when the cavity-precision $\beta$ is sufficiently large, the
hybrid distribution $h_{i}$ is strongly log-concave: its second log-derivative
is lower-bounded everywhere by a strictly positive quantity:
\[
-\frac{\partial^{2}}{\partial x^{2}}\log\left(h_{i}\left(x\right)\right)=\phi_{i}^{''}\left(x\right)+\beta\geq\phi_{i}^{''}\left(\mu_{0}\right)+\beta-B
\]

For log-concave distributions, the variance of any statistic is upper-bounded
by the Brascamp-Lieb inequality. In particular, the variance is upper-bounded
by:
\begin{equation}
\var_{h_{i}}\left(x\right)\leq E_{h_{i}}\left(\left[\phi_{i}^{''}\left(x\right)+\beta\right]^{-1}\right)\label{eq: Brascamp-Lieb upper-bounds the variance}
\end{equation}

From this and the curvature lower-bound, we get a coarse upper-bound
on the variance: 
\begin{eqnarray}
\var_{h_{i}}\left(x\right) & \leq & \left[\phi_{i}^{''}\left(\mu_{0}\right)+\beta-B\right]^{-1}\nonumber \\
 & \leq & \O\left(\beta^{-1}\right)\label{eq: coarse upper bound on the variance}
\end{eqnarray}

This coarse upper-bound is the first step of our proof.

Now that this coarse bound is established, we continue working on
$\var_{h_{i}}\left(x\right)$ and backtrack to eq. \ref{eq: Brascamp-Lieb upper-bounds the variance}
which we will simplify.

In order to simplify it, we will use several properties of the inverse
function: $M\rightarrow M^{-1}$ for $\linebreak M\in\left[\phi_{i}^{''}\left(\mu_{0}\right)+\beta-B,\phi_{i}^{''}\left(\mu_{0}\right)+\beta+B\right]$.

First, we can perform a simple Taylor expansion of the inverse function
around $E_{h_{i}}\left(\phi_{i}^{''}\left(x\right)\right)+\beta$
in order to simplify $\left[\phi_{i}^{''}\left(x\right)+\beta\right]^{-1}$:
\begin{eqnarray}
\left[\phi_{i}^{''}\left(x\right)+\beta\right]^{-1} & \approx & \ \!\left[E_{h_{i}}\left(\phi_{i}^{''}\left(x\right)\right)+\beta\right]^{-1}\nonumber \\
 &  & -\left[E_{h_{i}}\left(\phi_{i}^{''}\left(x\right)\right)+\beta\right]^{-2}\left(\phi_{i}^{''}\left(x\right)-E_{h_{i}}\left(\phi_{i}^{''}\left(x\right)\right)\right)\nonumber \\
 &  & +\left[E_{h_{i}}\left(\phi_{i}^{''}\left(x\right)\right)+\beta\right]^{-3}\left(\phi_{i}^{''}\left(x\right)-E_{h_{i}}\left(\phi_{i}^{''}\left(x\right)\right)\right)^{2}
\end{eqnarray}

This expression is a rough approximation with which it is hard to
do rigorous mathematics. We would like to replace it with an upper-bound.
Thankfully, the inverse function $M\rightarrow M^{-1}$ is convex
and its curvature is bounded over the interval of interest: $M\in\left[\phi_{i}^{''}\left(\mu_{0}\right)+\beta-B,\phi_{i}^{''}\left(\mu_{0}\right)+\beta+B\right]$.
We can then upper-bound it by a second degree Taylor expansion in
which we replace $\left[E_{h_{i}}\left(\phi_{i}^{''}\left(x\right)\right)+\beta\right]^{-3}$
by an upper-bound $\left[\phi_{i}^{''}\left(\mu_{0}\right)+\beta-B\right]^{-3}$.
This yields:
\begin{eqnarray}
\left[\phi_{i}^{''}\left(x\right)+\beta\right]^{-1} & \leq & \ \!\left[E_{h_{i}}\left(\phi_{i}^{''}\left(x\right)\right)+\beta\right]^{-1}\nonumber \\
 &  & -\left[E_{h_{i}}\left(\phi_{i}^{''}\left(x\right)\right)+\beta\right]^{-2}\left(\phi_{i}^{''}\left(x\right)-E_{h_{i}}\left(\phi_{i}^{''}\left(x\right)\right)\right)\nonumber \\
 &  & +\left[\phi_{i}^{''}\left(\mu_{0}\right)+\beta-B\right]^{-3}\left(\phi_{i}^{''}\left(x\right)-E_{h_{i}}\left(\phi_{i}^{''}\left(x\right)\right)\right)^{2}\label{eq: an airtight upper bound for the inverse curvature}
\end{eqnarray}

We can now get an upper-bound for the variance by combining this upper-bound
\eqref{eq: an airtight upper bound for the inverse curvature} and
the Brascamp-Lieb inequality eq. \eqref{eq: Brascamp-Lieb upper-bounds the variance}:
\begin{eqnarray}
\var_{h_{i}}\left(x\right) & \leq & E_{h_{i}}\left(\left[\phi_{i}^{''}\left(x\right)+\beta\right]^{-1}\right)\nonumber \\
 & \leq & \left[E_{h_{i}}\left(\phi_{i}^{''}\left(x\right)\right)+\beta\right]^{-1}+\left[\phi_{i}^{''}\left(\mu_{0}\right)+\beta-B\right]^{-3}\var_{h_{i}}\left(\phi_{i}^{''}\left(x\right)\right)\nonumber \\
 & \leq & \left[E_{h_{i}}\left(\phi_{i}^{''}\left(x\right)\right)+\beta\right]^{-1}+\left[\phi_{i}^{''}\left(\mu_{0}\right)+\beta-B\right]^{-3}K_{3}^{2}\var_{h_{i}}\left(x\right)\nonumber \\
 & \leq & \left[E_{h_{i}}\left(\phi_{i}^{''}\left(x\right)\right)+\beta\right]^{-1}+K_{3}^{2}\left[\phi_{i}^{''}\left(\mu_{0}\right)+\beta-B\right]^{-4}\label{eq: a less coarse upper-bound on the variance}
\end{eqnarray}
where we have bounded the variance of $\phi_{i}^{''}$ using a Taylor
expansion, and re-used our coarse bound on the variance \eqref{eq: coarse upper bound on the variance}.

We will now invert this expression \eqref{eq: a less coarse upper-bound on the variance}.We
use the convexity of the inverse function: $\left(a+b\right)^{-1}\geq a^{-1}-a^{-2}b$
(with $a=E_{h_{i}}\left(\phi_{i}^{''}\left(x\right)\right)+\beta$)
to find:
\begin{eqnarray}
\var_{h_{i}}^{-1}\left(x\right) & \geq & \left(\left[E_{h_{i}}\left(\phi_{i}^{''}\left(x\right)\right)+\beta\right]^{-1}+K_{3}^{2}\left[\phi_{i}^{''}\left(\mu_{0}\right)+\beta-B\right]^{-4}\right)^{-1}\nonumber \\
 & \geq & \left[E_{h_{i}}\left(\phi_{i}^{''}\left(x\right)\right)+\beta\right]-K_{3}^{2}\left[E_{h_{i}}\left(\phi_{i}^{''}\left(x\right)\right)+\beta\right]^{2}\left[\phi_{i}^{''}\left(\mu_{0}\right)+\beta-B\right]^{-4}\nonumber \\
 & \geq & \beta+E_{h_{i}}\left(\phi_{i}^{''}\left(x\right)\right)-K_{3}^{2}\frac{\left[\phi_{i}^{''}\left(\mu_{0}\right)+\beta+B\right]^{2}}{\left[\phi_{i}^{''}\left(\mu_{0}\right)+\beta-B\right]^{-4}}\nonumber \\
 & \geq & \beta+E_{h_{i}}\left(\phi_{i}^{''}\left(x\right)\right)-\O\left(K_{3}^{2}\beta^{-2}\right)
\end{eqnarray}

Note that in that last expression, the terms are ordered according
to their asymptotic behavior: linear-term, order 1 term (which also
contains a $\beta^{-1}$ term) and a $\beta^{-2}$ remainder.

We finally expand $\phi_{i}^{''}\left(x\right)$ around $\mu_{h_{i}}=E_{h_{i}}\left(x\right)$
in order to simplify the expected value $E_{h_{i}}\left(\phi_{i}^{''}\left(x\right)\right)$:
\begin{eqnarray*}
E_{h_{i}}\left(\phi_{i}^{''}\left(x\right)\right) & \geq & \phi_{i}^{''}\left(\mu_{h_{i}}\right)-\frac{K_{4}}{2}\var_{h_{i}}\left(x\right)\\
 & \geq & \phi_{i}^{''}\left(\mu_{h_{i}}\right)-\frac{K_{4}}{2}\left[\phi_{i}^{''}\left(\mu_{0}\right)+\beta-B\right]^{-1}\\
 & \geq & \phi_{i}^{''}\left(\mu_{h_{i}}\right)-\O\left(K_{4}\beta^{-1}\right)
\end{eqnarray*}

Which gives us a lower-bound for $\beta_{i}=\var_{h_{i}}\left(x\right)-\beta$:
\begin{eqnarray}
\beta_{i} & \geq & \phi_{i}^{''}\left(\mu_{h_{i}}\right)-K_{4}\left[\phi_{i}^{''}\left(\mu_{0}\right)+\beta-B\right]^{-1}-K_{3}^{2}\frac{\left[\phi_{i}^{''}\left(\mu_{0}\right)+\beta+B\right]^{2}}{\left[\phi_{i}^{''}\left(\mu_{0}\right)+\beta-B\right]^{-4}}\nonumber \\
 & \geq & \phi_{i}^{''}\left(\mu_{h_{i}}\right)-\O\left(K_{4}\beta^{-1}+K_{3}^{2}\beta^{-2}\right)\label{eq: a lower bound on beta_i which still refers to mu_h_i}
\end{eqnarray}

This is not our final lower-bound on $\beta_{i}$ because it refers
to the hybrid-mean $\mu_{h_{i}}$ and not to the cavity-mean $\mu_{0}$.
As explained in the proof outline, we will further along the way bound
$\mu_{h_{i}}-\mu_{0}$. We will then be able to backtrack to our bounds
on $\beta_{i}$ which refer to $\mu_{h_{i}}$ and, with yet another
Taylor expansion, make them use $\mu_{0}$ instead.

However, before we work on $\mu_{h_{i}}$, let's conclude our work
on $\beta_{i}$. We will now upper-bound $\beta_{i}$ which requires
lower-bounding $\var_{h_{i}}$.

We will make use of a Cramer-Rao inequality%
\footnote{Consider estimating $x$ from the observation $\hat{x}=x+\eta$ where
the noise has distribution $h_{i}\left(\eta-\mu_{h_{i}}\right)$:
it has mean 0 and variance $\var_{h_{i}}$.

$\hat{x}$ is an unbiased estimator with variance $\var_{h_{i}}$.
The Fisher information in $\hat{x}$ about $x$ is: $E_{h_{i}}\left(\phi_{i}^{''}\left(x\right)+\beta\right)$.
This gives the claimed Cramer-Rao inequality.

See example 10.22 from Saumard et al (2014).%
}. It reads:
\begin{eqnarray}
\var_{h_{i}}^{-1}\left(x\right) & \leq & E_{h_{i}}\left(\phi_{i}^{''}\left(x\right)+\beta\right)\nonumber \\
 & \leq & \beta+\phi_{i}^{''}\left(\mu_{h_{i}}\right)+K_{4}E_{h_{i}}\left(\left(x-\mu_{h_{i}}\right)^{2}/2\right)\nonumber \\
 & \leq & \beta+\phi_{i}^{''}\left(\mu_{h_{i}}\right)+\frac{K_{4}}{2}\var_{h_{i}}\left(x\right)\nonumber \\
 & \leq & \beta+\phi_{i}^{''}\left(\mu_{h_{i}}\right)+\frac{K_{4}}{2}\left[\phi_{i}^{''}\left(\mu_{0}\right)+\beta-B\right]^{-1}
\end{eqnarray}
where we have used a Taylor expansion and the coarse-bound on $\var_{h_{i}}\left(x\right)$
(eq. \eqref{eq: coarse upper bound on the variance}). 

This gives us an upper-bound on $\beta_{i}$:
\begin{eqnarray}
\beta_{i} & \leq & \phi_{i}^{''}\left(\mu_{h_{i}}\right)+\frac{K_{4}}{2}\left[\phi_{i}^{''}\left(\mu_{0}\right)+\beta-B\right]^{-1}\nonumber \\
 & \leq & \phi_{i}^{''}\left(\mu_{h_{i}}\right)+\O\left(K_{4}\beta^{-1}\right)\label{eq: an upper-bound on beta_i that still refers to mu_h_i}
\end{eqnarray}

Combining the upper and lower-bound (eqs. \ref{eq: a lower bound on beta_i which still refers to mu_h_i}
and \ref{eq: an upper-bound on beta_i that still refers to mu_h_i}),
we find that we can express the limit behavior of $\beta_{i}$ as
a function of $\mu_{h_{i}}$:
\begin{equation}
\beta_{i}=\phi_{i}^{''}\left(\mu_{h_{i}}\right)+\O\left(K_{4}\beta^{-1}+K_{3}^{2}\beta^{-2}\right)\label{eq: mu_h_i refering limit-behavior of beta_i}
\end{equation}

This last result concludes our first section on $\var_{h_{i}}\left(x\right)$
and $\beta_{i}$.%
\footnote{Note also the slight variant:
\begin{equation}
\beta_{i}=E_{h_{i}}\left(\phi_{i}^{''}\left(x\right)\right)+\O\left(K_{3}^{2}\beta^{-2}\right)\label{eq: expected-value limit behavior of beta_i}
\end{equation}

This variant is not used again in the rest of the work we present
here but is the next order of the expansion of $\beta_{i}$ which
could be of interest in expansions.%
}

\paragraph{Limit of $\mu_{h_{i}}$}

The next step in our proof is to work on the mean of the hybrid. We
will now show that $\mu_{h_{i}}\approx\mu_{0}$. This will give us
the final expression: $\beta_{i}\approx\phi_{i}^{''}\left(\mu_{0}\right)$
and will also be important in deriving the limit behavior of $r_{i}$.

We start from a ``Stein relationship'': with integration by parts,
we find that it would be true that, for any probability distribution,
the expected value of the log-gradient of the distribution is always
equal to 0:
\begin{equation}
E_{p}\left(\psi^{'}\left(x\right)\right)=0
\end{equation}

where we have used $p\left(x\right)\propto\exp\left(-\psi\left(x\right)\right)$
as an example.

We then apply that relationship to $h_{i}$ whose log-gradient has
contributions from the site and from the cavity distribution:
\begin{eqnarray}
E_{h_{i}}\left(\phi_{i}^{'}\left(x\right)+x\beta-\left(\beta\mu_{0}-\delta r\right)\right) & = & 0\label{eq: the Stein relationship}\\
\beta\left[\mu_{h_{i}}-\mu_{0}\right] & = & -\delta r-E_{h_{i}}\left(\phi_{i}^{'}\left(x\right)\right)\label{eq: Slightly rephrasing the Stein relationship}
\end{eqnarray}

We now perform a Taylor expansion of $E_{h_{i}}\left(\phi_{i}^{'}\left(x\right)\right)$:
\begin{equation}
E_{h_{i}}\left(\phi_{i}^{'}\left(x\right)\right)\approx\phi_{i}^{'}\left(\mu_{0}\right)+\phi_{i}^{''}\left(\mu_{0}\right)\left[\mu_{h_{i}}-\mu_{0}\right]\label{eq: a Taylor expansion}
\end{equation}

The error in that expression can be upper-bounded:
\begin{eqnarray}
\text{error} & \leq & \frac{K_{3}}{2}\var_{h_{i}}\left(x\right)\nonumber \\
 & \leq & \frac{K_{3}}{2}\left[\phi_{i}^{''}\left(\mu_{0}\right)+\beta-B\right]^{-1}\nonumber \\
 & \leq & \O\left(K_{3}\beta^{-1}\right)\label{eq: bounding the error of the Taylor expansion}
\end{eqnarray}

We combine that Taylor expansion \eqref{eq: a Taylor expansion} and
that bound \eqref{eq: bounding the error of the Taylor expansion}
with the Stein relationship \eqref{eq: Slightly rephrasing the Stein relationship}:
\begin{equation}
\left[\beta+\phi_{i}^{''}\left(\mu_{0}\right)\right]\left[\mu_{h_{i}}-\mu_{0}\right]=-\delta r-\phi_{i}^{'}\left(\mu_{0}\right)+\O\left(K_{3}\beta^{-1}\right)\label{eq: almost bounding mu_h_i - mu_0}
\end{equation}

From which we can deduce a coarse bound on $\mu_{h_{i}}-\mu_{0}$:
\begin{eqnarray}
\mu_{h_{i}}-\mu_{0} & = & \left[\beta+\phi_{i}^{''}\left(\mu_{0}\right)\right]^{-1}\left[-\delta r-\phi_{i}^{'}\left(\mu_{0}\right)\right]+\O\left(K_{3}\beta^{-2}\right)\nonumber \\
 & = & \O\left(\left[\delta r+\phi_{i}^{'}\left(\mu_{0}\right)\right]\beta^{-1}+K_{3}\beta^{-2}\right)\label{eq: limit behavior of mu_h_i - mu_0}
\end{eqnarray}

This coarse bound will now be used to slightly update eq. \eqref{eq: a lower bound on beta_i which still refers to mu_h_i}
in order to remove the dependency on $\mu_{h_{i}}$ and obtain the
final equation on the limit behavior of $\beta_{i}$. We will also
use it in the final section to compute the limit behavior of $r_{i}$.

\paragraph*{Returning to $\beta_{i}$}

We combine the coarse bound on $\mu_{h_{i}}$ with our expression
for the limit behavior of $\beta_{i}$. We find:
\begin{eqnarray}
\beta_{i}-\phi_{i}^{''}\left(\mu_{0}\right) & = & \phi_{i}^{''}\left(\mu_{h_{i}}\right)-\phi_{i}^{''}\left(\mu_{0}\right)+\O\left(K_{4}\beta^{-1}+K_{3}^{2}\beta^{-2}\right)\nonumber \\
 & = & \O\left(K_{3}\left[\mu_{h_{i}}-\mu_{0}\right]+K_{4}\beta^{-1}+K_{3}^{2}\beta^{-2}\right)\nonumber \\
 & = & \O\left(K_{3}\left[\delta r+\phi_{i}^{'}\left(\mu_{0}\right)\right]\beta^{-1}+K_{3}^{2}\beta^{-2}+K_{4}\beta^{-1}+K_{3}^{2}\beta^{-2}\right)\label{eq:limit behavior of beta_i}
\end{eqnarray}

Which completely concludes the proof for the expression of the limit
behavior of $\beta_{i}$.

\paragraph{Limit of $r_{i}=\protect\var_{h_{i}}^{-1}\left(x\right)\mu_{h_{i}}-\left(\beta\mu_{0}-\delta r\right)$}

We now turn to the task of approximating $r_{i}$. Our very first
step here is to use the Stein relationship we already used above (eq.
\eqref{eq: the Stein relationship}). This gives a simple expression
for $r_{i}$ without going through the limit behavior of the hybrid
mean and the hybrid variance. 

We start from the Stein relationship we had before (eq. \eqref{eq: the Stein relationship}):
\begin{eqnarray}
E_{h_{i}}\left(\phi_{i}^{'}\left(x\right)+x\beta-\left(\beta\mu_{0}-\delta r\right)\right) & = & 0\nonumber \\
E_{h_{i}}\left(\phi_{i}^{'}\left(x\right)\right)+\beta\mu_{h_{i}}-\beta\mu_{0}+\delta r & = & 0\label{eq: a different rewritting of the Stein relationship}
\end{eqnarray}

Let $g\left(x\right)$ be the Gaussian with same mean and variance
as the hybrid. The natural parameters of $g$ are the sum of the natural
parameters of the cavity distribution and of the approximation of
$h_{i}$. In other words, its natural parameters are $\beta+\beta_{i}$
and $\beta_{i}\mu_{0}-\delta_{r}+r_{i}$:
\[
g=\mathcal{P}\left(h_{i}\right)\propto\exp\left(-\left[\beta+\beta_{i}\right]\frac{x^{2}}{2}+\left[\beta\mu_{0}-\delta r+r_{i}\right]x\right)
\]
If we now apply the Stein relationship to $g$, we find:
\begin{eqnarray}
E_{g}\left(\left[\beta+\beta_{i}\right]x-\left(\beta\mu_{0}-\delta r+r_{i}\right)\right) & = & 0\nonumber \\
\left(\beta_{i}\mu_{h_{i}}-r_{i}\right)+\left(\beta\mu_{h_{i}}-\beta\mu_{0}+\delta r\right) & = & 0\label{eq: A different rewritting for the Gaussian projection of the hybrid}
\end{eqnarray}

Combining these two equations \eqref{eq: a different rewritting of the Stein relationship},
\eqref{eq: A different rewritting for the Gaussian projection of the hybrid},
we find:
\begin{equation}
r_{i}=\beta_{i}\mu_{h_{i}}-E_{h_{i}}\left(\phi_{i}^{'}\left(x\right)\right)\label{eq: exact nice expression for r_i}
\end{equation}
which is a nice simple expression for $r_{i}$. We have thus found
a way to work directly on $r_{i}$ while avoiding direct work on the
mean and variance.%
\footnote{Note that, in the exact same way as for eq. \eqref{eq: expected-value limit behavior of beta_i},
this equation is interesting in its own right for expansions of our
result as it is the next order of the expression for $r_{i}$.%
}

We will now simplify eq. \eqref{eq: exact nice expression for r_i}.
We start from a slight rewriting:
\begin{equation}
r_{i}=\beta_{i}\mu_{0}+\beta_{i}\left(\mu_{h_{i}}-\mu_{0}\right)-E_{h_{i}}\left(\phi_{i}^{'}\left(x\right)\right)
\end{equation}

Let's expand the $E_{h_{i}}\left(\phi_{i}^{'}\left(x\right)\right)$
term. We just perform a Taylor expansion (around $\mu_{0}$):
\begin{eqnarray}
E_{h_{i}}\left(\phi_{i}^{'}\left(x\right)\right) & = & \phi_{i}^{'}\left(\mu_{0}\right)+\phi_{i}^{''}\left(\mu_{0}\right)\left(\mu_{h_{i}}-\mu_{0}\right)+\O\left(K_{3}\left(\var_{h_{i}}+\left(\mu_{h_{i}}-\mu_{0}\right)^{2}\right)\right)\nonumber \\
 & = & \phi_{i}^{'}\left(\mu_{0}\right)+\phi_{i}^{''}\left(\mu_{0}\right)\left(\mu_{h_{i}}-\mu_{0}\right)+\O\left(K_{3}\beta^{-1}+K_{3}\left(\mu_{h_{i}}-\mu_{0}\right)^{2}\right)
\end{eqnarray}

This gives us the following expression for $r_{i}$:
\begin{equation}
r_{i}\approx\beta_{i}\mu_{0}-\phi_{i}^{'}\left(\mu_{0}\right)+\left(\beta_{i}-\phi_{i}^{''}\left(\mu_{0}\right)\right)\left(\mu_{h_{i}}-\mu_{0}\right)
\end{equation}

We now just need to bound $\beta_{i}-\phi_{i}^{''}\left(\mu_{0}\right)$
to finalize our proof, which we have already done in eq. \eqref{eq:limit behavior of beta_i}.
We then obtain the final expression for the limit behavior of $r_{i}$:
\begin{eqnarray}
r_{i}-\beta_{i}\mu_{0}-\phi_{i}^{'}\left(\mu_{0}\right) & = & \O\left(\left(\beta_{i}-\phi_{i}^{''}\left(\mu_{0}\right)\right)\left(\mu_{h_{i}}-\mu_{0}\right)+K_{3}\beta^{-1}+K_{3}\left(\mu_{h_{i}}-\mu_{0}\right)^{2}\right)\nonumber \\
 & = & \O\left(K_{3}\beta^{-1}+K_{4}\beta^{-1}\left(\mu_{h_{i}}-\mu_{0}\right)+K_{3}\left(\mu_{h_{i}}-\mu_{0}\right)^{2}\right)\label{eq: final expression for r_i}
\end{eqnarray}

\end{proof}

\subsection{Limit behavior in high-dimensions\label{sub:Limit-behavior-in high-dimensions}}

Our result generalizes to the p-dimensional case: when the probability
distribution we are trying to approximate concerns a p-dimensional
random variable.

The main difficulty in stating and in understanding that case comes
from the tensor notation that we have to work with. Indeed, all of
the moments and the derivatives we have to work with shift from being
scalars to being p-dimensional tensors of various orders.

\subsubsection{Tensor notation and rephrasing the assumptions}

Let's first recall what tensors are. A tensor of order $k$ is simply
a multilinear mapping of $\left(\mathbb{R}^{p}\right)^{k}$ to $\mathbb{R}$:
it takes $k$ vectors and returns a scalar. We will note this $T\left[v_{1},v_{2}\dots v_{k}\right]$.
This is a simple extension of vectors (order 1 tensors) and matrices
(order 2 tensors). In our examples, we will always deal with symmetric
tensors for which the order of the arguments does not influence the
outputted value. Finally, an order $k$ tensor can be also be used
as a multilinear mapping with fewer than $l<k$ entries: it then returns
an order $k-l$ tensor. This simply corresponds to specifying some
of the inputs to the original tensor and leaving the assignation of
the other inputs for latter. We will note by leaving to be specified
inputs with a minus sign: eg $T\left[v_{1}\dots v_{k-2},-,-\right]$
for $l=2$. For example, if we were to perform a Taylor expansion
of $\nabla\phi\left(\mathbf{x}\right)$ , this Taylor expansion would
need to return a vector (ie: an order 1 tensor). The first term of
the expansion would be $H\phi_{i}\left(\mathbf{x}_{0}\right)\left[\left(\mathbf{x}-\mathbf{x}_{0}\right),-\right]$,
the second term $\phi_{i}^{\left(3\right)}\left[\left(\mathbf{x}-\mathbf{x}_{0}\right),\left(\mathbf{x}-\mathbf{x}_{0}\right),-\right]$,
etc.

In order to state our result in $p$-dimensions, we will first need
to extend our assumptions to the high-dimensional case. This is relatively
easy for the boundedness condition on $H\phi_{i}$: we just need to
find a matrix $\mathbf{B}$ such that:
\begin{equation}
\forall\mathbf{x}_{1},\mathbf{x}_{2}\ \ H\phi_{i}\left(\mathbf{x}_{1}\right)-H\phi_{i}\left(\mathbf{x}_{2}\right)\leq\mathbf{B}\label{eq: matrix valued bounded range of second derivative}
\end{equation}
where the order relationship is the standard (semi) order between
symmetric matrices: $\mathbf{B_{1}}\geq\mathbf{B}_{2}$ if their difference
is positive semi-definite. This captures the idea that the range of
the second derivatives is small.

The boundedness condition on the higher-derivatives requires us to
define a norm on tensors. We will simply use the norm induced on tensors
by the $L_{2}$ norm on vectors. This is defined, for a tensor of
order $k$ by:
\begin{equation}
\left\Vert T\right\Vert =\max_{v_{1}\dots v_{k}}\frac{T\left[v_{1}\dots v_{k}\right]}{\prod\left\Vert v_{i}\right\Vert _{2}}
\end{equation}

For vectors (who are order 1 tensors), this is of course the $L_{2}$
norm. For matrices, this corresponds to the maximum eigenvalue. 

This induced norm has good behavior when we input an order $k$ tensor
$T_{k}$ with only $l$ inputs: the resulting order $k-l$ tensor
$T_{k-l}=T_{k}\left[v_{1}\dots v_{l},-,\dots,-\right]$ has bounded
norm:
\begin{equation}
\left\Vert T_{k-l}\right\Vert \leq\left\Vert T_{k}\right\Vert \prod_{i=1}^{l}\left\Vert v_{i}\right\Vert _{2}
\end{equation}

This simply follows from the definition of the norm. We will use this
fact several times because we often need to perform Taylor expansion
of gradient vectors or of Hessian matrices. For example, if we perform
an order 0 Taylor expansion of $H\phi_{i}\left(\mathbf{x}\right)$
around $\bmu_{0}$, we can use this result to bound the matrix-valued
reminder term:
\begin{eqnarray*}
R_{3}\left(\mathbf{x}\right) & = & H\phi_{i}\left(\mathbf{x}\right)-H\phi_{i}\left(\bmu_{0}\right)\\
\left\Vert R_{3}\left(\mathbf{x}\right)\right\Vert  & \leq & K_{3}\left\Vert \mathbf{x}-\bmu_{0}\right\Vert 
\end{eqnarray*}

Once we have defined this norm, we can define the infinity norm for
a tensor-valued function like the third derivative and the fourth
derivative. We can then state our second assumption in p-dimensions:
\begin{eqnarray}
\left\Vert \phi_{i}^{\left(3\right)}\right\Vert _{\infty} & \leq & K_{3}\\
\left\Vert \phi_{i}^{\left(4\right)}\right\Vert _{\infty} & \leq & K_{4}\label{eq: matrix valued bounded third and fourth}
\end{eqnarray}

\subsubsection{Updating the proof}

Let's now see how our proof changes now that we are in the high-dimensional
case.

\paragraph{Summary of the differences}

All the important steps of the proof are identical: the Brascamp-Lieb
theorem, the Cramer-Rao bound, the Stein's method trick all work in
high-dimensions. Thus, the algebra of the proof is the same. What
does change is the bounding of the error terms: throughout the proof,
we repeatedly bound $E_{h_{i}}\left(\left(x-\mu_{h_{i}}\right)^{2}\right)$
which corresponds in high-dimensions to $E_{h_{i}}\left(\left\Vert \mathbf{x}-\bmu_{h_{i}}\right\Vert ^{2}\right)$.
In 1D, we bounded this using the coarse bound on the variance of the
hybrid we got from the Brascamp-Lieb theorem:
\[
\var_{h_{i}}\left(x\right)\leq\left[H\phi_{i}\left(\mu_{0}\right)+\beta-B\right]^{-1}
\]

We will do the exact same thing in high-dimensions. Let's note $\mathbf{Q}$
the precision matrix of the cavity distribution. The coarse variance
bound reads:
\[
\cov_{h_{i}}\left(\mathbf{x}\right)\leq\left[H\phi_{i}\left(\bmu_{0}\right)+\mathbf{Q}-\mathbf{B}\right]^{-1}
\]

which gives:
\begin{eqnarray*}
E_{h_{i}}\left(\left\Vert \mathbf{x}-\bmu_{h_{i}}\right\Vert ^{2}\right) & \leq & Tr\left(\left[H\phi_{i}\left(\bmu_{0}\right)+\mathbf{Q}-\mathbf{B}\right]^{-1}\right)\\
 & \leq & p\left\Vert \left[H\phi_{i}\left(\bmu_{0}\right)+\mathbf{Q}-\mathbf{B}\right]^{-1}\right\Vert 
\end{eqnarray*}

which gives us the effect of the dimension parameter $p$ on our results:
it scales the error (at most) linearly.

\paragraph{An example: going through the bound on $\mathbf{Q}_{i}$}

Now that we have given a high-level explanation of the difference
between the two cases, let's detail how the proof of the result needs
to be extended for the bound on $\mathbf{Q}_{i}$.

In the 1D proof, we bounded: $\var_{h_{i}}\left(H\phi_{i}\left(x\right)\right)$
using a Taylor expansion. In high-dimensions, this gives:

\begin{eqnarray*}
\left\Vert H\phi_{i}\left(\mathbf{x}\right)-H\phi_{i}\left(\bmu_{h_{i}}\right)\right\Vert  & \leq & K_{3}\left\Vert \mathbf{x}-\bmu_{h_{i}}\right\Vert \\
\left\Vert H\phi_{i}\left(\mathbf{x}\right)-H\phi_{i}\left(\bmu_{h_{i}}\right)\right\Vert ^{2} & \leq & K_{3}^{2}\left\Vert \mathbf{x}-\bmu_{h_{i}}\right\Vert ^{2}
\end{eqnarray*}

We then have:
\begin{eqnarray*}
\var_{h_{i}}\left(H\phi_{i}\left(x\right)\right) & \leq & K_{3}^{2}E_{h_{i}}\left(\left\Vert \mathbf{x}-\bmu_{h_{i}}\right\Vert ^{2}\right)\\
 & \leq & K_{3}^{2}Tr\left(\cov_{h_{i}}\left(\mathbf{x}\right)\right)\\
 & \leq & K_{3}^{2}Tr\left(\left[H\phi_{i}\left(\bmu_{0}\right)+\mathbf{Q}-\mathbf{B}\right]^{-1}\right)\\
 & \leq & pK_{3}^{2}\left\Vert \left[H\phi_{i}\left(\bmu_{0}\right)+\mathbf{Q}-\mathbf{B}\right]^{-1}\right\Vert 
\end{eqnarray*}
where we have used the coarse bound on the covariance of the hybrid
distribution that we obtained from the Brascamp-Lieb theorem.

\paragraph{The high-dimensional theorem}

This then leads us to the following theorem stating our limit result
in high-dimensions. Note that the ``high precision'' limit here
means that \emph{all }eigenvalues of $\mathbf{Q}$ should go to infinity. 
\begin{thm}
High-dimensional extension

Th. \ref{thm:Limit-behavior-of the site approximation-1} also applies
when approximating a target distribution over a p-dimensional space.
The small term is then 
\begin{eqnarray}
\epsilon & = & \max\left(K_{3},K_{4}\right)Tr\left(\left[H\phi_{i}\left(\bmu_{0}\right)+\mathbf{Q}-\mathbf{B}\right]^{-1}\right)\\
 & \leq & p\max\left(K_{3},K_{4}\right)\left\Vert \left[H\phi_{i}\left(\bmu_{0}\right)+\mathbf{Q}-\mathbf{B}\right]^{-1}\right\Vert 
\end{eqnarray}
.\end{thm}
\begin{rem*}
In this result, the error scales linearly with the dimensionality
of the space over which the target distribution is expressed. This
means that, in the non-parametric case for which $p=n$, our results
stop providing useful bounds, and it wouldn't be possible to prove
that EP is asymptotically exact. We believe that this is a weakness
of our proof and not of EP, and that a more careful bounding of the
error would show that EP still behaves properly in non-parametric
problems.
\end{rem*}

\subsection{Limit behavior under weaker assumptions\label{sub:Limit-behavior-under weaker assumptions}}

Our choice of constraints in the main text is restrictive: while it
can be applied for wide classes of statistical models, it's still
quite far from being a necessary condition. In this section, we expose
an alternative set of assumptions such that a variant of theorem \ref{thm:Limit-behavior-of the site approximation-1}
holds.

In this section, we will simply use a global bounding of the site:
\begin{equation}
l_{i}\left(x\right)\leq1\label{eq: weaker assumption, bounded site}
\end{equation}
We will also replace our global assumption constraining the higher
log-derivatives of the sites with a purely local constraint: we assume
that there exists a compact neighborhood $\mathcal{I}$ of $\mu_{0}$
such that $\phi_{i}$ is derivable four times, and that those four
derivatives are bounded inside $\mathcal{I}$.
\begin{thm}
Global bounding of the site (eq. \ref{eq: weaker assumption, bounded site})
and local smoothness constraints are sufficient for th. \ref{thm:Limit-behavior-of the site approximation-1}
to hold\end{thm}
\begin{proof}
The key idea behind this proof is the following: we will show that
under our assumptions the sequence of probability densities $h_{i}\left(x|\beta\right)$
converges towards their restrictions to $\mathcal{I}$. This convergence
is strong enough to imply convergence of the parameters of the Gaussian
approximation. Since the restriction of the site to $\mathcal{I}$:
$l_{i}\left(x\right)1\left(x\in\mathcal{I}\right)$ fulfills the hypotheses
of the weaker theorem, we get the claimed result.

\paragraph*{Convergence of $h_{i}\left(x|\beta\right)$}

As a first step, let's slightly rewrite the ``cavity'' Gaussian
by centering it around its limit mean $\mu_{0}$:
\begin{eqnarray*}
q_{-i} & \propto & \exp\left(-\frac{\beta}{2}x^{2}+(\beta\mu_{0}-\delta r)x\right)\\
 & \propto & \exp\left(-\frac{\beta}{2}\left(x-\mu_{0}\right)^{2}-\delta r\left(x-\mu_{0}\right)\right)
\end{eqnarray*}

Now consider the sequence of probability density functions: $\newline h_{i}\left(x|\beta\right)=\frac{l_{i}\left(x\right)}{l_{i}\left(\mu_{0}\right)}\sqrt{\frac{\beta}{2\pi}}\exp\left(-\frac{\beta}{2}\left(x-\mu_{0}\right)^{2}-\delta r\left(x-\mu_{0}\right)\right)$.
Note that these are unproperly normalized. However, in the limit,
the sequence of integrals $\int h_{i}\left(x|\beta\right)\rightarrow1$
as we will now show.

Note $M_{\beta}\left(t\right)=\int\exp\left(t\sqrt{\beta}\left(x-\mu_{0}\right)\right)h_{i}\left(x|\beta\right)dx$:
this is the moment generating function of the random variable $y_{\beta}=\sqrt{\beta}\left(x_{\beta}-\mu_{0}\right)$
(though, once again, note that this is only an asymptotically properly
normalized density as our proof that $M_{\beta}\left(0\right)\rightarrow1$
will show). For any $\beta>0$, $M_{\beta}\left(t\right)$ is finite
(the quadratic terms in the log dominate and $l_{i}\left(x\right)\leq1$).

Let's prove that $M_{\beta}\left(t\right)$ converges pointwise for
any value of $t$ to $\exp\left(-\frac{t^{2}}{2}\right)$. This will
prove that $y_{\beta}$ converges to a Gaussian of variance $1$,
and that its density is asymptotically properly normalized.

Fix $t$. Let $\epsilon>0$. Since the derivatives are bounded on
$\mathcal{I}$, there must exist a ball centered on $\mu_{0}$ : $B\left(\mu_{0},r_{\epsilon}\right)\subset\mathcal{I}$,
where $\left|l_{i}\left(x\right)-l\left(\mu_{0}\right)\right|\leq\epsilon$.
Let's decompose the integral in $M_{\beta}$ into the integral over
$B$ (where that simple bound holds) and the one over $\mathbb{R}/B$.

For the central region:
\begin{eqnarray*}
\mbox{cent}\left(\beta\right) & = & \int_{B}\exp\left(\sqrt{\beta}t\left(x-\mu_{0}\right)\right)h_{i}\left(x|\beta\right)\\
 & \geq & \frac{l_{i}\left(\mu_{0}\right)-\epsilon}{l_{i}\left(\mu_{0}\right)}\sqrt{\frac{\beta}{2\pi}}\int_{B}\exp\left(\sqrt{\beta}t\left(x-\mu_{0}\right)-\frac{\beta}{2}\left(x-\mu_{0}\right)^{2}-\delta r\left(x-\mu_{0}\right)\right)\\
 & \geq & \frac{l_{i}\left(\mu_{0}\right)-\epsilon}{l_{i}\left(\mu_{0}\right)}\sqrt{\frac{\beta}{2\pi}}\int_{B}\exp\left(\frac{\beta}{2}\left(x-\mu_{0}-\frac{t}{\sqrt{\beta}}\right)^{2}+\frac{\left(t-\delta r/\sqrt{\beta}\right)^{2}}{2}\right)
\end{eqnarray*}
and a similar upper-bound.

As $\beta\rightarrow\infty$, all of the mass of $\sqrt{\beta}\exp\left(\frac{\beta}{2}\left(x-\mu_{0}-\frac{t}{\sqrt{\beta}}\right)^{2}\right)$
becomes concentrated inside $B$. Thus, the bounds converge to $\frac{l_{i}\left(\mu_{0}\right)\pm\epsilon}{l_{i}\left(\mu_{0}\right)}\exp\left(\frac{t^{2}}{2}\right)$
as $\beta\rightarrow\infty$.

For the exterior region:
\begin{eqnarray*}
\mbox{ext}\left(\beta\right) & = & \int_{x\notin B}\exp\left(\sqrt{\beta}t\left(x-\mu_{0}\right)\right)h_{i}\left(x|\beta\right)\\
 & \leq & \frac{\sqrt{\beta}}{l_{i}\left(\mu_{0}\right)}\int_{x\notin B}\exp\left(\sqrt{\beta}t\left(x-\mu_{0}\right)-\frac{\beta}{2}\left(x-\mu_{0}\right)^{2}-\delta r\left(x-\mu_{0}\right)\right)\\
 & \leq & \frac{\sqrt{\beta}}{l_{i}\left(\mu_{0}\right)}\int_{x\notin B}\exp\left(\frac{\beta}{2}\left(x-\mu_{0}-\frac{t}{\sqrt{\beta}}\right)^{2}+\frac{\left(t-\delta r/\sqrt{\beta}\right)^{2}}{2}\right)
\end{eqnarray*}
which converges to $0$ as $\beta\rightarrow\infty$ (exponentially
fast), since all of the probability mass of $\sqrt{\beta}\exp\left(\frac{\beta}{2}\left(x-\mu_{0}-\frac{t}{\sqrt{\beta}}\right)^{2}\right)$
is concentrated inside $B$.

With those convergences, we could thus find $\beta_{0}$ such that
$\forall\beta\geq\beta_{0}$, 
\begin{equation}
\left|M_{\beta}\left(t\right)-\exp\left(-\frac{t^{2}}{2}\right)\right|\leq3\epsilon\exp\left(\frac{t^{2}}{2}\right)\left(l_{i}\left(\mu_{0}\right)\right)^{-1}
\end{equation}
This proves that $M_{\beta}\left(t\right)$ converges pointwise to
$\exp\left(\frac{t^{2}}{2}\right)$.

$M_{\beta}\left(t\right)$ is the moment generating function of an
unnormalized probability distribution. The moment generating function
for the normalized variable is simply found by taking the ratio against
$M_{\beta}\left(0\right)$. Since $M_{\beta}\left(0\right)$, the
MGF of the normalized density also converges to $\exp\left(\frac{t^{2}}{2}\right)$:
\[
\mbox{lim}_{\beta\rightarrow\infty}\frac{M_{\beta}\left(t\right)}{M_{\beta}\left(0\right)}\rightarrow\exp\left(\frac{t^{2}}{2}\right)
\]

Note that this means that the sequence of random variables $y_{\beta}$
converges in MGF to a Gaussian centered at $0$ and with variance
$1$. Convergence in MGF implies weak-convergence and convergence
of all moments, and implies the convergence of $x_{\beta}$ to a Gaussian
centered at $\mu_{0}$ and with variance $\beta^{-1}$. However, this
is a secondary point: we will have to work directly on $M_{\beta}\left(t\right)$
in order to derive our next results.

\paragraph*{Convergence of $h_{i}\left(x|\beta\right)$ to its restriction to
$\mathcal{I}$}

We will use this MGF convergence to prove that $h_{i}\left(x|\beta\right)$
converges to its restriction to $\mathcal{I}$: $h_{i}^{r}\left(x|\beta\right)=h_{i}\left(x|\beta\right)1\left(x\in\mathcal{I}\right)$,
and furthermore that the error we make in this approximation are negligible
compared to the limit behavior of interest in $\mu_{h_{i}}$ and in
$\var_{h_{i}}$.

Let's first prove the convergence of the variances. $h_{i}^{r}$ fulfills
the hypotheses of the less-general theorem \ref{thm:Limit-behavior-of the site approximation-1},
so that:
\[
\var_{h_{i}^{r}}\left(x\right)\approx\beta^{-1}+\beta^{-2}\phi_{i}^{''}\left(\mu_{0}\right)+\O\left(\beta^{-3}\right)
\]
The important feature is that the deviation of interest from the $\beta^{-1}$
limit is of order $\beta^{-2}$.

Let's compute the variance for the unrestricted hybrid $h_{i}$. We
can decompose the variance into the contribution from $\mathcal{I}$
and the contribution from $\mathbb{R\backslash\mathcal{I}}$. The
contribution from the outer region can be upper-bounded by the following
argument:
\begin{itemize}
\item Let $t_{0}$ such that $\exp\left(-\left|t_{0}\left(x-\mu_{0}\right)\right|\right)\left(x-\mu_{0}\right)^{2}$
is strictly decreasing on the outer region $\mathbb{R}\backslash\mathcal{I}$.
\item Let's recall the Markov-bound: it proves that a positive random-variable
with finite mean can not have arbitrarily large deviations. It can
be extended to prove that a positive random-variable with finite $k^{th}$
moment must also have bounded $k_{2}^{th}$ moments for all $k_{2}<k$.\\
Because $\exp\left(-\left|t_{0}\left(x-\mu_{0}\right)\right|\right)\left(x-\mu_{0}\right)^{2}$
is decreasing, we can find a Markov-like bound on the contribution
of the outer region of $h_{i}$ to the variance. Note $\rho$ the
radius of $\mathcal{I}$ and consider the following bound on $\left(x-\mu_{0}\right)^{2}1\left(x\notin\mathcal{I}\right)$:
\begin{eqnarray*}
\exp\left(-\left|t_{0}\left(x-\mu_{0}\right)\right|\right)\left(x-\mu_{0}\right)^{2}1\left(x\notin\mathcal{I}\right) & \leq & \exp\left(-\left|t_{0}\rho\right|\right)\rho^{2}\\
\left(x-\mu_{0}\right)^{2}1\left(x\notin\mathcal{I}\right) & \leq & \exp\left(\left|t_{0}\left(x-\mu_{0}\right)\right|\right)\exp\left(-\left|t_{0}\rho\right|\right)\rho^{2}
\end{eqnarray*}
When we take the expected value under $h_{i}$, we obtain the Markov-like
bound:
\begin{eqnarray*}
E_{h_{i}}\left(\left(x-\mu_{0}\right)^{2}1\left(x\notin\mathcal{I}\right)\right) & \leq & E_{h_{i}}\left(\exp\left(\left|t_{0}\left(x-\mu_{0}\right)\right|\right)\right)\exp\left(-\left|t_{0}\rho\right|\right)\rho^{2}\\
 & \leq & \frac{1}{M_{\beta}\left(0\right)}\left(M_{\beta}\left(-\frac{t_{0}}{\sqrt{\beta}}\right)+M_{\beta}\left(\frac{t_{0}}{\sqrt{\beta}}\right)\right)\exp\left(-\left|t_{0}\rho\right|\right)\rho^{2}
\end{eqnarray*}

\item Now consider adapting that last bound to the value of $\beta$ with
$t_{\beta}=\sqrt{\beta}t_{0}$ (note that for $t\geq t_{0}$, the
corresponding function is still decreasing outside of $\mathcal{I}$).
We get:
\[
E_{h_{i}}\left(\left(x-\mu_{0}\right)^{2}1\left(x\notin\mathcal{I}\right)\right)\leq\frac{1}{M_{\beta}\left(0\right)}\left(M_{\beta}\left(-t_{0}\right)+M_{\beta}\left(t_{0}\right)\right)\exp\left(-\left|\sqrt{\beta}t_{0}\rho\right|\right)\rho^{2}
\]
and since the $M_{\beta}\left(\pm t_{0}\right)$ converge to $\exp\left(\frac{t_{0}^{2}}{2}\right)$,
we have that, for sufficiently high $\beta$:
\[
E_{h_{i}}\left(\left(x-\mu_{0}\right)^{2}1\left(x\notin\mathcal{I}\right)\right)\leq\left(2\exp\left(\frac{t_{0}^{2}}{2}\right)+\epsilon\right)\exp\left(-\left|\sqrt{\beta}t_{0}\rho\right|\right)\rho^{2}
\]

\end{itemize}
This proves that the contribution of the outer region to the variance
decreases exponentially in $\sqrt{\beta}$ which is much faster than
the error in the central region which is in $\beta^{-2}$.

The error in the central region is thus found to dominate the other
error terms so that $\var_{h_{i}}^{-1}-\beta\rightarrow\var_{h_{i}^{r}}^{-1}-\beta$.

By a similar argument, we can prove that the deviation of the mean
from $\mu_{0}$ is also dominated by the error in the central region,
and that the error contributed by the outer region decays exponentially.
This is done by finding $t_{0}$ such that $\linebreak\exp\left(-\left|t_{0}\left(x-\mu_{0}\right)\right|\right)\left|x-\mu_{0}\right|$.
This proves that: $\var_{h_{i}}^{-1}E_{h_{i}}\left(x\right)-\beta\mu_{0}\rightarrow\var_{h_{i}^{r}}^{-1}E_{h_{i}^{r}}\left(x\right)-\beta\mu_{0}$
exponentially fast, while the size of the central term is of order
1.

We thus have that the Gaussian approximations of the $h_{i}\left(x|\beta\right)$
converges towards the Gaussian approximations of $h_{i}^{r}\left(x|\beta\right)$.
Since we can apply the weaker theorem to $h_{i}^{r}\left(x|\beta\right)$,
we get the claimed result. 
\end{proof}

\subsection{Detailled proof of the stable region theorem}

In this section, we detail the proof of the main text theorem on stable
regions of aEP and EP.
\begin{thm}
Convergence of fixed points of EP and aEP\label{thm:Convergence-of-fixed-points of EP and aEP-1}

There exists an EP and an aEP fixed point close to the CGA of $p\left(x\right)$
at $x^{\star}$ if $\phi_{i}^{''}\left(x^{\star}\right)$ is sufficiently
large. More precisely, if:
\begin{eqnarray*}
\delta_{aEP} & = & \max\left(K_{3},K_{4}\right)\sum\left|\phi_{i}^{'}\left(x^{\star}\right)\right|\left[\psi^{''}\left(x^{\star}\right)\right]^{-1}\\
\delta & = & n\max\left(K_{3},K_{4}\right)\left[\psi^{''}\left(x^{\star}\right)\right]^{-1}
\end{eqnarray*}
 are order 1 quantities and $\psi^{''}\left(x^{\star}\right)$ is
large, then the limit of the stable regions on the global approximation,
$n\Delta_{r}$ and $n\Delta_{\beta}$, scale as $\mathcal{O}\left(\delta_{aEP}+\delta\right)$
for aEP and as $\mathcal{O\left(\delta\right)}$ for EP.\end{thm}
\begin{proof}
Let's start with aEP. Let's assume that we start from region of the
parameter space with the following form:

\begin{eqnarray*}
\left|r_{aEP}-\beta_{aEP}x^{\star}\right| & \leq & \gamma\left(\delta+\delta_{aEP}\right)\\
\left|\beta_{aEP}-\psi^{''}\left(x^{\star}\right)\right| & \leq & \gamma\left(\delta+\delta_{aEP}\right)
\end{eqnarray*}

The critical feature here is that the stable region has size of order
1 and that $\psi^{''}\left(x^{\star}\right)$ is large.

We will now apply th. \ref{thm:Limit-behavior-of the site approximation-1}
with $\mu_{0}=x^{\star}$ to bound the next parameter values:
\begin{eqnarray}
r_{aEP}^{new} & = & \beta_{aEP}^{new}x^{\star}+\O\left(n\max\left(K_{3},K_{4}\right)\left[\psi^{''}\left(x^{\star}\right)\right]^{-1}\right)\nonumber \\
\beta_{aEP}^{new} & = & \psi^{''}\left(x^{\star}\right)+\O\left(n\max\left(K_{3},K_{4}\right)\left[\psi^{''}\left(x^{\star}\right)\right]^{-1}\right)\nonumber \\
 &  & +\O\left(\max\left(K_{3,}K_{4}\right)\left[\left|r_{aEP}-\beta_{aEP}x^{\star}\right|+\sum_{i=1}^{n}\left|\phi_{i}^{'}\left(x^{\star}\right)\right|\right]\left[\psi^{''}\left(x^{\star}\right)\right]^{-1}\right)
\end{eqnarray}

In those equations, we have used the fact that $\O\left(\left[\frac{n-1}{n}\psi^{''}\left(x^{\star}\right)-\left|\beta_{aEP}-\psi^{''}\left(x^{\star}\right)\right|\right]^{-1}\right)=\O\left(\left[\psi^{''}\left(x^{\star}\right)\right]^{-1}\right)$.
Note how the order 1 deviation: $\left|\beta_{aEP}-\psi^{''}\left(x^{\star}\right)\right|$
does not affect the limit behavior.

Similarly, the order 1 deviation in $r$: $\left|r_{aEP}-\beta_{aEP}x^{\star}\right|$,
is similarly ``asymptotically silent'': $\left|r_{aEP}-\beta_{aEP}x^{\star}\right|\left[\psi^{''}\left(x^{\star}\right)\right]^{-1}$
is negligible before $\sum_{i=1}^{n}\left|\phi_{i}^{'}\left(x^{\star}\right)\right|\left[\psi^{''}\left(x^{\star}\right)\right]^{-1}$.

The limit behavior thus simplifies into:

\begin{eqnarray*}
r_{aEP} & = & \beta_{aEP}x^{\star}+\O\left(\delta\right)\\
\beta_{aEP} & = & \psi^{''}\left(x^{\star}\right)+\O\left(\delta+\delta_{aEP}\right)
\end{eqnarray*}

When $\psi^{''}\left(x^{\star}\right)$ is sufficiently large, there
exists $c$ such that: $\left|r_{aEP}-\beta_{aEP}x^{\star}\right|\leq c\delta$.
If $\gamma\geq c$, then this proves that the aEP iteration is contractive
for the linear-shift natural parameter.

Similarly, when $\psi^{''}\left(x^{\star}\right)$ is sufficiently
large, $\left|\beta_{aEP}-\psi^{''}\left(x^{\star}\right)\right|\leq c_{2}\left(\delta+\delta_{aEP}\right)$,
and if $\gamma\geq\max\left(c,c_{2}\right)$, the aEP iteration is
contractive in both directions which proves the result.

The result for EP follows from the exact same reasoning.
\end{proof}

\subsection{Exactness of aEP and EP in the large-data limit}

In this final section, we will prove our only non-deterministic result
concerning the behavior of EP and aEP. We will prove that, if we accumulate
sites which are generated according to some process guaranteeing Local
Asymptotic Normality of the log-posterior, then aEP and EP have a
fixed point which converges to the CGA in total-variation. Furthermore,
if the process generating the sites respects some assumptions which
constrain the mass of the posterior outside of a close neighborhood
of the highest-mode, then the posterior $p_{n}\left(x\right)$ and
its CGA $q_{n}\left(x\right)$ converge towards one another in total-variation.
This ensures that aEP and EP both have a fixed point which is asymptotically
exact in total-variation.

We were unable to obtain a Bernstein-von Mises results which offers
exactly the conditions we needed, so we we had to re-derive every
result from scratch. We make no claim for originality for this section
which is extremely similar to many other asymptotic studies of likelihoods
and posterior distributions.

Let's first detail the assumptions that we will need on the random
process generating the sites $l_{i}\left(x\right)$. We will assume
that:
\begin{itemize}
\item the $l_{i}$ are i.i.d
\item Their distribution is such that the expected value function:
\[
x\rightarrow E\left(\phi_{i}\left(x\right)\right)
\]
has a global maximum at $x_{0}$. We will note $I_{0}=E\left(\phi_{i}^{''}\left(x_{0}\right)\right)>0$.
\item their distribution is such that the following quantities are finite:
\begin{eqnarray*}
\var\left(\phi_{i}^{''}\left(x_{0}\right)\right) & < & \infty\\
\var\left(\phi_{i}^{'}\left(x_{0}\right)\right) & < & \infty\\
E\left(\left|\phi_{i}^{'}\left(x_{0}\right)\right|\right) & < & \infty
\end{eqnarray*}
Their size controls how regular the process generating the sites is.
Note that, in the well-specified case, $\var\left(\phi_{i}^{'}\left(x_{0}\right)\right)=E\left(\phi_{i}^{''}\left(x_{0}\right)\right)=I_{0}$
by integration by parts, in which case $I_{0}$ is the Fisher information
provided by the observations. 
\end{itemize}
These conditions will ensure that the posterior is Locally Asymptotically
Normal so that it can be approximated locally by a Gaussian. Furthermore,
in order to have a global approximation, we will require that the
model is such that estimation with it is consistent: with probability
tending to 1, it converges towards $x_{0}$ which is the parameter
that fits the data best. The condition is that: $\forall\epsilon>0$,
the random variables $\int p_{n}\left(x\right)1\left(\left|x-x_{0}\right|\leq\epsilon\right)dx$
converge to 1 in probability as $n\rightarrow\infty$.

Armed with these assumptions, let's now derive several results on
the log-posterior.
\begin{lem}
As $n\rightarrow\infty$, the log-curvature at $x_{0}$ grows linearly.
More precisely, $\forall\epsilon>0$:\label{lem: log-curvature grows linearly}
\[
\sum_{i=1}^{n}\phi_{i}^{''}\left(x_{0}\right)\geq n(I_{0}-\epsilon)
\]
with probability tending to 1. Similarly, the log-gradient at $x_{0}$
is of order $\sqrt{n}$. With probability tending to 1, we have that:
\[
\left|\sum_{i=1}^{n}\phi_{i}^{'}\left(x_{0}\right)\right|\leq\sqrt{n\left[\var\left(\phi_{i}^{'}\left(x_{0}\right)\right)+\epsilon\right]}
\]
\end{lem}
\begin{proof}
This result is simply a large-data limit concentration result.

The mean of the cumulative sum is $nI_{0}$ while its variance is
$n\var\left(\phi_{i}^{''}\left(x_{0}\right)\right)$. By Chebyshev's
concentration theorem, the probability that the cumulative sum deviates
from its mean decays at speed $\sqrt{n}$.

By the exact same argument, we get that the cumulative log-gradient
is small: its mean is 0, by definition of $x_{0}$, and its variance
is $n\var\left(\phi_{i}^{'}\left(x_{0}\right)\right)$. Once again,
by applying Chebyshev's theorem, we get the claimed result.\end{proof}
\begin{lem}
As $n\rightarrow\infty$, there exists with probability tending to
1 a mode $x_{n}^{\star}$ in close proximity to $x_{0}$. More precisely:
\[
x_{0}-x_{n}^{\star}=\O\left(n^{-1/2}\right)
\]

\end{lem}
Throughout the following lemmas, we will always refer to $x_{n}^{\star}$
without explicitly mentionning that $x_{n}^{\star}$ does not always
exist (but has probability tending to 1 of existing).
\begin{proof}
With lemma \ref{lem: log-curvature grows linearly}, we have that
with probability 1, the log-curvature at $x_{0}$ grows linearly and
the log-gradient is of order $\sqrt{n}$. Assume for the following
that the log-gradient at $x_{0}$ is negative. 

We also have that the third log-derivative is bounded:
\[
\left|\sum_{i=1}^{n}\phi_{i}^{\left(3\right)}\left(x\right)\right|\leq nK_{3}
\]

Let's then consider the following Taylor lower-bound to the log-gradient:
\[
\sum_{i=1}^{n}\phi_{i}^{'}\left(x\right)=\sum_{i=1}^{n}\phi_{i}^{'}\left(x_{0}\right)+\left(x-x_{0}\right)\left(\sum_{i=1}^{n}\phi_{i}^{''}\left(x\right)\right)-\frac{\left(x-x_{0}\right)^{2}}{2}nK_{3}
\]

If the log-curvature is sufficiently large, and the log-gradient sufficiently
small, the second degree polynomial in $x$ has two roots. More precisely,
this happens when the discriminant is positive:
\[
Disc=\left(\sum_{i=1}^{n}\phi_{i}^{''}\left(x\right)\right)^{2}-4\left|\sum_{i=1}^{n}\phi_{i}^{'}\left(x_{0}\right)\right|nK_{3}
\]

Given the growth rates of the log-gradient and of the log-curvature,
we have that, with probability 1, this discriminant is positive, because
it grows asymptotically as $\sqrt{n}$. Thus, in the large-data limit,
we are guaranteed to have a root in close vicinity of $x_{0}$. Furthermore,
it is easy to check that the dominating term in $x_{n}^{\star}$ is
the root of the order 1 polynomial $\sum_{i=1}^{n}\phi_{i}^{'}\left(x_{0}\right)+\left(x-x_{0}\right)\left(\sum_{i=1}^{n}\phi_{i}^{''}\left(x\right)\right)=0$:
\[
x_{n}^{\star}-x_{0}=-\left(\sum_{i=1}^{n}\phi_{i}^{''}\left(x\right)\right)^{-1}\left(\sum_{i=1}^{n}\phi_{i}^{'}\left(x_{0}\right)\right)+\O\left(n^{-1}\right)
\]

This leading term is of order $n^{-1/2}$ which concludes this proof%
\footnote{This proof is straightforward to extend to the high-dimensional case.
One simply needs to diagonalize the $\left(\sum_{i=1}^{n}\phi_{i}^{''}\left(x\right)\right)$
matrix. %
}.\end{proof}
\begin{lem}
The log-curvature at $x_{n}^{\star}$ grows linearly with probability
tending to 1. More precisely, $\forall\epsilon>0$:
\[
\sum_{i=1}^{n}\phi_{i}^{''}\left(x_{n}^{\star}\right)\geq n(I_{0}-\epsilon)
\]
with probability tending to 1.

Similarly, the following quantity grows linearly:
\[
\sum_{i=1}^{n}\left|\phi_{i}^{'}\left(x_{n}^{\star}\right)\right|\approx nE\left(\left|\phi_{i}^{'}\left(x_{0}\right)\right|\right)
\]
\end{lem}
\begin{proof}
We compute the log-curvature at $x_{n}^{\star}$ from the log-curvature
at $x_{0}$ with a Taylor expansion:
\[
\left|\sum_{i=1}^{n}\phi_{i}^{''}\left(x_{n}^{\star}\right)-\phi_{i}^{''}\left(x_{0}\right)\right|\leq nK_{3}\left|x_{n}^{\star}-x_{0}\right|
\]

Since $\left|x_{n}^{\star}-x_{0}\right|$ is of order $n^{-1/2}$,
we get that the error between the two is of order $\sqrt{n}$ and
that the log-curvature at $x_{n}^{\star}$ grows linearly.

Similarly:
\[
\left|\sum_{i=1}^{n}\left|\phi_{i}^{'}\left(x_{n}^{\star}\right)\right|-\sum_{i=1}^{n}\left|\phi_{i}^{'}\left(x_{0}\right)\right|\right|\leq\left(\sum_{i=1}^{n}\left|\phi_{i}^{''}\left(x_{0}\right)\right|\right)\left|x_{n}^{\star}-x_{0}\right|
\]

and we have that this quantity also grows linearly.
\end{proof}
With this result, we reach a turning point in our proof. So far, we
have proved that $x_{n}^{\star}$ exits and that various quantities
measured at $x_{n}^{\star}$ grow linearly (with probability tending
to 1). This proves that the log-posterior is Locally Asymptotically
Normal (LAN) around $x_{n}^{\star}$ (with probability tending to
1). 

We will now show that the LAN behavior of the posterior ensures that
aEP and EP have a stable region in a small neighborhood around $x_{n}^{\star}$
and that all global approximations of $p_{n}\left(x\right)$ inside
that stable region converge towards $q_{n}\left(x\right)$: the CGA
of $p_{n}\left(x\right)$ at $x_{n}^{\star}$.
\begin{lem}
We can apply main text th. 3 to $x_{n}^{\star}$ with probability
tending to 1.\label{lem: there exists a stable region - LEMMA}

Thus, there exists a stable region of order 0:
\begin{eqnarray*}
n\Delta_{n,r} & = & \O\left(n^{0}\right)\\
n\Delta_{n,\beta} & = & \O\left(n^{0}\right)
\end{eqnarray*}
around the CGA at $x_{n}^{\star}$.\end{lem}
\begin{proof}
The log-curvature at the mode $x_{n}^{\star}$ grows linearly in $n$
and so does $n\max\left(K_{3},K_{4}\right)$. Similarly, for aEP,
$\sum_{i=1}^{n}\left|\phi_{i}^{'}\left(x_{n}^{\star}\right)\right|$
grows linearly. Application of the theorem then gives that there exists
a stable region in the large-data limit.\end{proof}
\begin{lem}
\label{lem: Convergence of fixed-points of aEP and EP to the CGA}All
Gaussian approximations of $p_{n}\left(x\right)$ inside the stable
region of lemma \ref{lem: there exists a stable region - LEMMA} converge
in KL divergence to $q_{n}\left(x\right)$: the CGA centered at $x_{n}^{\star}$.
This convergence in KL divergence implies a convergence in total-variation. 

More precisely:$\forall r,\beta$ such that:
\begin{eqnarray*}
\left|r-\beta x_{n}^{\star}\right| & \leq & n\Delta_{n,r}\\
\left|\beta-\sum_{i=1}^{n}\phi_{i}^{''}\left(x_{n}^{\star}\right)\right| & \leq & n\Delta_{n,\beta}
\end{eqnarray*}

then:
\begin{eqnarray*}
KL\left(\N\left(x|r,\beta\right),q_{n}\left(x\right)\right) & = & \O\left(n^{-1}\right)\\
d_{TV}\left(\N\left(x|r,\beta\right),q_{n}\left(x\right)\right) & = & \O\left(n^{-1/2}\right)
\end{eqnarray*}
\end{lem}
\begin{proof}
The formula for the KL divergence between two Gaussian distributions
(parameterized with mean and variance) $q_{1}=\N\left(\mu_{1},\beta_{1}^{-1}\right)$
and $q_{2}=\N\left(\mu_{2},\beta_{2}^{-1}\right)$ is the following:
\begin{eqnarray*}
2KL\left(q_{1},q_{2}\right) & = & \beta_{2}\left(\mu_{1}-\mu_{2}\right)^{2}+\frac{\left(\beta_{2}-\beta_{1}\right)}{\beta_{1}}-\log\left(\beta_{2}/\beta_{1}\right)\\
 & \approx & \beta_{2}\left(\mu_{1}-\mu_{2}\right)^{2}+\frac{\left(\beta_{2}-\beta_{1}\right)^{2}}{2\beta_{1}^{2}}
\end{eqnarray*}
where the approximation is valid when $\mbox{\ensuremath{\frac{\left(\beta_{2}-\beta_{1}\right)}{\beta_{1}}\approx}0}$.%
\footnote{In the high-dimensional case, the KL divergence is:
\[
2KL\left(q_{1},q_{2}\right)=\left(\bmu_{1}-\bmu_{2}\right)\mathbf{Q}_{2}\left(\bmu_{1}-\bmu_{2}\right)+Tr\left(\mathbf{Q}_{2}\mathbf{Q}_{1}^{-1}\right)-d-\log\left(\frac{\left|\mathbf{Q}_{2}\right|}{\left|\mathbf{Q}_{1}\right|}\right)
\]

If $\mathbf{Q}_{1}$ and $\mathbf{Q}_{2}$ are co-diagonal, the complicated
term $Tr\left(\mathbf{Q}_{2}\mathbf{Q}_{1}^{-1}\right)-d-\log\left(\frac{\left|\mathbf{Q}_{2}\right|}{\left|\mathbf{Q}_{1}\right|}\right)$
is, like the 1D case, almost quadratic. Even when they are not co-diagonal:
\[
\log\left(\left|\mathbf{Q}_{2}+\Delta\right|\right)\approx Tr\left(\mathbf{Q}_{1}^{-1}\Delta+\frac{1}{2}\mathbf{Q}_{1}^{-1}\Delta\mathbf{Q}_{1}^{-1}\right)
\]

so that the quadratic approximation is still valid:
\[
2KL\left(q_{1},q_{2}\right)\approx\left(\bmu_{1}-\bmu_{2}\right)\mathbf{Q}_{2}\left(\bmu_{1}-\bmu_{2}\right)+\frac{1}{2}Tr\left(\left[\mathbf{Q}_{2}-\mathbf{Q}_{1}\right]^{2}\mathbf{Q}_{1}^{-2}\right)
\]
}

When we apply this expression to our distributions of interest, we
find the following upper-bound for the KL divergence between any Gaussian
approximation in the stable region and the CGA at $x_{n}^{\star}$
(slightly abusively, we use the quadratic limit approximation of the
$\beta$ dependent term):
\begin{eqnarray*}
2KL & \leq & \left(\sum_{i=1}^{n}\phi_{i}^{''}\left(x_{n}^{\star}\right)\right)\left(\frac{n\Delta_{n,r}}{\sum_{i=1}^{n}\phi_{i}^{''}\left(x_{n}^{\star}\right)-n\Delta_{n,\beta}}\right)^{2}+\frac{1}{2}\left(\frac{n\Delta_{n,\beta}}{\sum_{i=1}^{n}\phi_{i}^{''}\left(x_{n}^{\star}\right)-n\Delta_{n,\beta}}\right)^{2}\\
2KL & \leq & \O\left(n^{-1}+n^{-2}\right)\\
2KL & \leq & \O\left(n^{-1}\right)
\end{eqnarray*}

We finally apply Pinsker's inequality, which relates the KL divergence
to the total-variation metric: $KL\geq2\left(d_{TV}\right)^{2}$ and
we get the claimed result.\end{proof}
\begin{lem}
\label{lem:(Van-de-Vaart)}\citet{KleijnVanDerVaart:BernsteinVonMisesUnderMisspec}
If $p_{n}\left(x\right)$ is Locally Asymptotically Normal and its
mass concentrates near $x_{n}^{\star}$, then:
\[
d_{TV}\left(p_{n}\left(x\right),q_{n}\left(x\right)\right)=O\left(n^{-1/2}\right)
\]
\end{lem}
\begin{proof}
We just sketch the proof of Kleijn et al.

Fix $\epsilon>0$. Under our assumption, with probability 1, the mass
of $p_{n}\left(x\right)$ concentrates inside the ball $\left|x-x_{0}\right|\leq\epsilon$
as $n\rightarrow\infty$.

By a simple Taylor expansion inside the ball around $x_{n}^{\star}$,
we find that the CGA $q_{n}\left(x\right)$ is a good approximation
of $p\left(x\right)$. Some care must be taken to ensure the rate
of convergence holds.
\end{proof}
With these last two lemmas \ref{lem: Convergence of fixed-points of aEP and EP to the CGA}
and \ref{lem:(Van-de-Vaart)}, we conclude our proof of the theorem:
under our conditions, all aEP and EP fixed points near the CGA at
$x_{n}^{\star}$ converge in total-variation (with speed $\O\left(n^{-1/2}\right)$)
to the true posterior $p_{n}\left(x\right)$.

\subsection{Exactness on probit and logit regression}

Let's now show that both for probit and logit regression, EP is exact.

In both cases, the likelihoods $l_{i}\left(\mathbf{x}\right)$ have
the following simple form:
\[
l_{i}\left(\mathbf{x}\right)=a\left(\mathbf{v}_{i}^{t}\mathbf{x}\right)
\]
where $\mathbf{v}_{i}$ is the vector of predictors for the $i^{th}$
datapoint and $a$ is the link function (or activation function) of
the model.

Both for the probit and the logit case, $a$ is a log-concave function
with bounded derivatives of all orders. Thus, as long as the predictor
vectors $\mathbf{v}_{i}$ are guaranteed to be bounded, then the $l_{i}$
are guaranteed to respect our assumptions on the derivatives (eqs.
\ref{eq: assumption 1, bound the range of the log-curvature} and
\ref{eq: assumption 2, bound the 3 and 4 derivatives}). Furthermore,
the log-concavity of $a$ makes it so that we can get rid of the identifiability
assumption, as identifiability is implied by the other assumptions.
\begin{lem}
If the Fisher information matrix if strictly positive $\mathbf{I}_{0}>0$,
then the posterior distribution is identifiable.\end{lem}
\begin{proof}
the posterior distribution is a product of log-concave distributions.
It is thus log-concave.

Furthermore, for the posterior, the log-curvature at $x_{n}^{\star}$
is growing linearly almost at speed $\mathbf{I}_{0}$ and the log
third derivative is bounded by $nK_{3}$.

The case with the larger tails that fits this picture is somewhat
like the Huber log-likelihood (quadratic center with linear tails):
it has a central curved region surrounded by linear tails. More precisely,
the log-curvature of the posterior is larger than:
\begin{equation}
\sum_{i=1}^{n}\phi_{i}^{''}\left(\mathbf{x}\right)\geq\max\left(\sum_{i=1}^{n}\phi_{i}^{''}\left(\mathbf{x}_{n}^{\star}\right)-nK_{3}\left\Vert \mathbf{x}-\mathbf{x}_{n}^{\star}\right\Vert ,0\right)
\end{equation}

We can integrate the limit behavior of the log-curvature in this expression
yielding:
\begin{equation}
\sum_{i=1}^{n}\phi_{i}^{''}\left(\mathbf{x}\right)\geq n\max\left(\mathbf{I}_{0}-K_{3}\left\Vert \mathbf{x}-\mathbf{x}_{n}^{\star}\right\Vert ,0\right)
\end{equation}

Critically, the log-curvature grows (approximately) linearly with
$n$. If we simply integrate the above expressoin twice, this remains
true for the negative log-posterior:

\begin{equation}
\sum_{i=1}^{n}\left(\phi_{i}\left(\mathbf{x}\right)-\phi_{i}\left(\mathbf{x}_{n}^{\star}\right)\right)\propto n\label{eq:approximate growth}
\end{equation}

In order to conclude, we will use a brand new result: theorem III.1
of \citet{Pereyra2016}. Pereyra shows that, for a log-concave probability
distribution, most of the mass is concentrated in a region where the
negative log-posterior is not too high above the minimum. Eq. \ref{eq:approximate growth}
shows that the negative log-posterior grows linearly with the number
of data-points $n$. By applying Pereyra's theorem, we have that the
probability concentrates around $\mathbf{x}_{n}^{\star}$ which concludes
our proof.
\end{proof}
A probit or logit regression thus respects all of our hypotheses.
Thus, in the large-data limit, both aEP and EP are exact on a probit
or logit model, as long as the Fisher information matrix is strictly
positive.

\section{aEP: an alternative to EP ?}

In the main text, we used aEP as a theoretical tool to study the asymptotics
of EP, but could it hold practical interest as well? aEP is simpler
than standard EP, and in particular it requires fewer matrix factorizations.
In standard EP computing the covariance matrix of cavity distributions
is a $\mathcal{O}\left(nm^{3}\right)$ or $\mathcal{O}\left(nm^{2}\right)$
operation (where $m$ is the number of parameters), depending on the
problem and the implementation. In aEP the covariance of the cavity
is just $\frac{n}{n-1}\bS$, the current covariance matrix, so that
forming the cavity distribution is a very cheap $\mathcal{O}\left(m^{2}\right)$
operation. Depending on $m$ aEP may be substantially faster. 

Another advantage of aEP is that, due to its much smaller parameter
set ($\mathcal{O}\left(m^{2}\right)$ vs $\mathcal{O}\left(nm^{2}\right)$),
generic numerical tools for fixed point iterations may be used out-of-the-box.
We experimented with R package SQUAREM \citep{Varadhan:SimpleGloballyConvMethods},
a set of algorithms that seek to accelerate fixed point iterations.
Our limited experimentation indicates that although SQUAREM does not
always achieve speed-ups, the fact that step sizes are chosen makes
aEP very robust. 

On the other hand, since aEP is an asymptotic approximation of EP,
we will incur a loss in performance. We ran some simulations to see
how different aEP and EP are in practical examples. We picked two
statistical models, Cauchy regression and probit regression. Probit
regression is a well-know EP success story \citep{KussRasmussen:AssessingApproxInfGP,NickishRasmussen:ApproxGaussianProcClass},
with well-behaved, log-concave sites. Cauchy regression features non-log
concave sites and posterior distributions may be multimodal.

Probit regression is the following model:

\begin{eqnarray*}
y_{i} & = & \mbox{sign}\left(\mathbf{x}_{i}^{t}\bm{\alpha}+\epsilon\right)\\
\epsilon & \sim & \N\left(0,1\right)
\end{eqnarray*}

where $\mathbf{x}_{i}$ is a vector of covariates, while Cauchy regression
is:

\begin{eqnarray*}
y_{i} & = & \mathbf{x}_{i}^{t}\bm{\alpha}+\epsilon\\
\epsilon & \sim & Cauchy\left(0,1\right)
\end{eqnarray*}

in both cases inference is for the regression coefficients $\bm{\alpha}$.
We used the standard factorization of the posterior (over likelihood
sites) with hybrid moments computed numerically. The prior over $\bm{\alpha}$
was set in both cases to $\bm{\alpha}\sim\N\left(0,1\right)$. The
regressors $\mathbf{x}_{i}\in\mathbb{R}^{4}$ were B-spline functions
evaluated on a grid of \emph{n} locations over the unit interval.
Data were generated according to the model. We ran aEP and EP for
20 passes at speed $\gamma=0.4$, since the Cauchy likelihood induced
occasional convergence problems. 

To measure the difference between aEP and EP, we used a relative difference
in means:

\begin{equation}
d_{\mu}=\frac{\left|\mu_{EP}-\mu_{aEP}\right|}{\mbox{min}\left(\sigma_{EP},\sigma_{aEP}\right)}\label{eq:rel_diff_means}
\end{equation}

which expresses how different the estimates are in units of standard
deviations. Differences in estimated posterior variance was summarized
by a ratio:

\begin{equation}
d_{\sigma}=\mbox{max}\left(\frac{\sigma_{EP}}{\sigma_{aEP}},\frac{\sigma_{aEP}}{\sigma_{EP}}\right)\label{eq:sd_ratio}
\end{equation}

Both measures were averaged over the $m=4$ parameters.

The results are shown on fig. \ref{fig:aEP-vs-EP}. aEP and EP are
both exact in large $n$, but the differences between the Cauchy and
the probit model are notable (one order of magnitude). In the probit
model aEP and EP are practically the same with just 20 datapoints,
with relative differences in means reaching a maximum of 5\%, whereas
differences in the Cauchy model can reach 40\%. With enough datapoints
the differences disappear in the Cauchy model as well.

\begin{center}
\begin{figure}
\begin{centering}
\includegraphics[width=10cm]{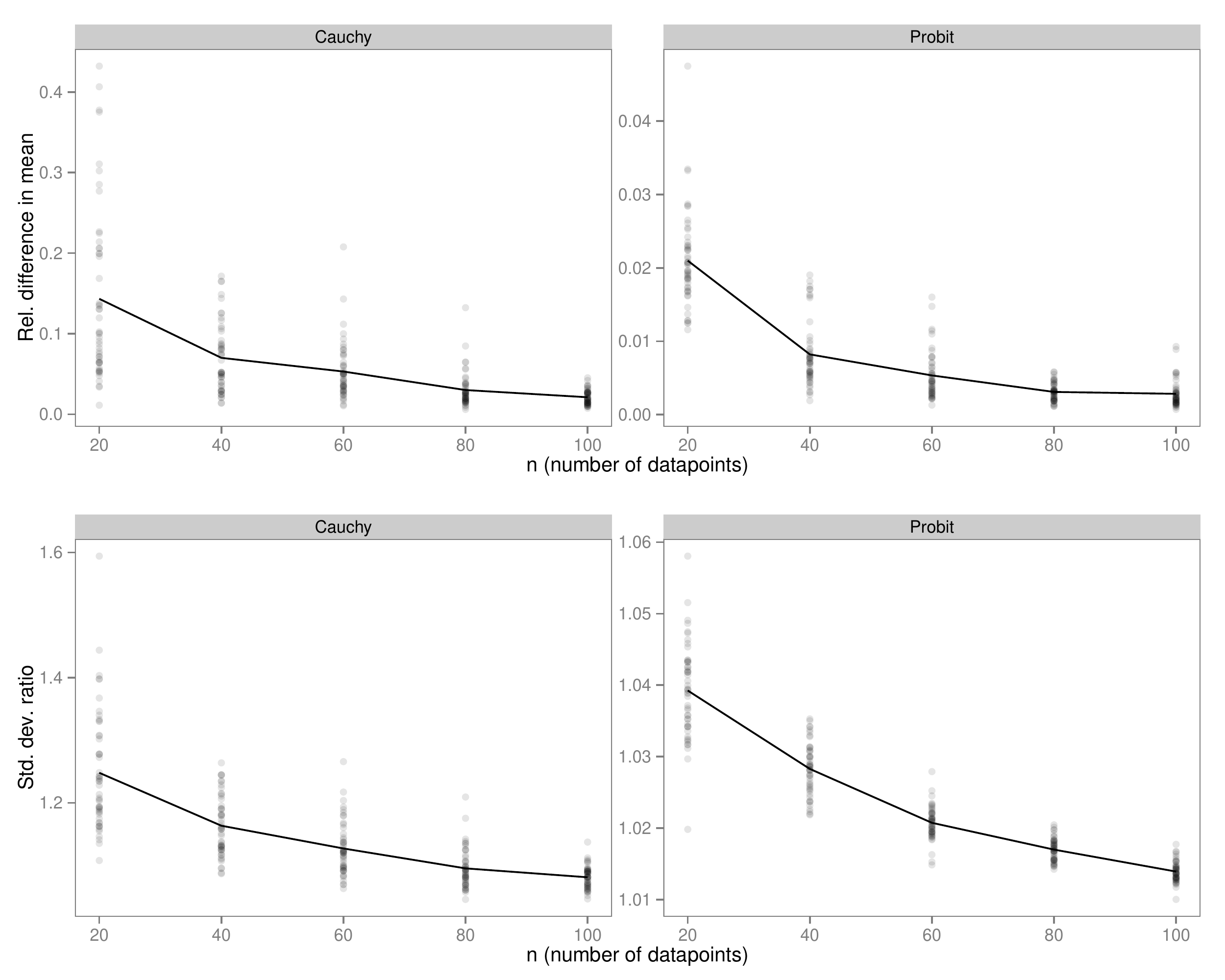}
\par\end{centering}

\protect\caption{aEP vs EP in probit and Cauchy models. Note that the vertical scales
are not the same across panels. Upper row: relative difference in
means (eq. \eqref{eq:rel_diff_means}) between aEP and EP. Lower row:
ratio of standard deviations between aEP and EP ($d_{\sigma}$ in
\eqref{eq:sd_ratio}). Dots represent individual simulations, the
continuous line connects the means. \label{fig:aEP-vs-EP}}
\end{figure}

\par\end{center}

When applying aEP, one must be careful to remember that it uses the
approximation that $\forall i,\ \lambda_{i}\approx\frac{1}{N}\sum\lambda_{j}$.
This means that if a few sites are outliers and have a very exceptional
contribution to the posterior, then they might make this approximation
wrong and aEP might give a very poor approximation. If applying aEP,
it thus seems sensible to check the validity of the assumption at
the last point of the iteration as a simple sanity check to verify
the quality of the approximation.

One interesting possibility is to run aEP until it gets close to a
fixed point to take advantage of the smaller amounts of computations,
and then switch to the EP iteration to take advantage of the possible
increased precision of the EP algorithm.
\end{document}